\documentclass[a4paper,UKenglish,cleveref, autoref, numberwithinsect, thm-restate]{lipics-v2021}

\title{Kernelization for $H$-Packing Revisited} 

\author{Tomohiro Koana}{Graduate School of Information Science and Technology, The University of Tokyo, Japan}{tomohiro.koana@gmail.com}{https://orcid.org/0000-0002-8684-0611}{Supported in part by JST CREST Grant Number JPMJCR24Q2 and JST ERATO Grant Number JPMJER2301.}

\author{Soh Kumabe}{CyberAgent, Japan}{kumabe\_soh@cyberagent.co.jp}{https://orcid.org/0000-0002-1021-8922
}{}

\authorrunning{T. Koana and S. Kumabe} 

\Copyright{Tomohiro Koana and Soh Kumabe} 



\category{} 

\relatedversion{} 

\nolinenumbers 

\usepackage[utf8]{inputenc}
\usepackage{aliascnt}
\usepackage{enumerate}
\usepackage{amsmath}
\usepackage{amsfonts}
\usepackage{nicefrac}
\usepackage{amssymb}
\usepackage{amsthm}
\usepackage{mathtools}
\usepackage{tikz}
\usetikzlibrary{calc}
\usepackage{pgfplots}
\pgfplotsset{compat=newest}
\usepackage{pgfplotstable}
\usepackage{xcolor}
\usepackage{tikz}
\usepackage{thmtools}
\usepackage{todonotes}

\theoremstyle{definition}
\newtheorem{rrule}{Rule}

\crefname{rrule}{Rule}{Rules}
\Crefname{theorem}{Theorem}{Theorems}
\Crefname{lemma}{Lemma}{Lemmas}
\Crefname{proposition}{Proposition}{Propositions}
\crefname{subfigure}{Figure}{Figures}
\Crefname{subfigure}{Figure}{Figures}

\newcommand{\problemdef}[3]{%
  \par\medskip
  \textsc{#1}\par
  \textbf{Input:} #2\par
  \textbf{Question:} #3\par
  \medskip
}

\newif\ifshortversion
\providecommand{\shortversionbuild}{false}
\expandafter\ifx\csname shortversion\shortversionbuild\endcsname\relax
  \shortversiontrue
\else
  \csname shortversion\shortversionbuild\endcsname
\fi
\long\def\shortonly#1{%
  \ifshortversion
    #1%
  \fi
}
\long\def\longonly#1{%
  \ifshortversion
  \else
    #1%
  \fi
}

\newif\ifproofomittednext
\proofomittednextfalse
\newcommand{\proofomittedsymbol}{\ensuremath{\bigstar}}
\newcommand{\proofomittedmark}{\textup{(\proofomittedsymbol)}}
\newcommand{\markproofomitted}{\global\proofomittednexttrue}

\makeatletter
\def\proofomitted@headmark{%
  \ifshortversion
    \ifproofomittednext
      \nobreakspace\proofomittedmark
    \fi
  \fi
  \global\proofomittednextfalse
}
\def\@begintheorem#1#2[#3]{%
  \deferred@thm@head{\the\thm@headfont \thm@indent
    \@ifempty{#1}{\let\thmname\@gobble}{\let\thmname\@iden}%
    \@ifempty{#2}{\let\thmnumber\@gobble}{\let\thmnumber\@iden}%
    \@ifempty{#3}{\let\thmnote\@gobble}{\let\thmnote\@iden}%
    \thm@swap\swappedhead\thmhead{#1}{#2}{#3}%
    \proofomitted@headmark
    \the\thm@headpunct
    \thmheadnl
    \hskip\thm@headsep
  }%
  \ignorespaces}
\makeatother

\DeclarePairedDelimiterX{\abs}[1]{\lvert}{\rvert}{#1}
\DeclarePairedDelimiterX{\norm}[1]{\lVert}{\rVert}{#1}
\DeclarePairedDelimiterX{\ceil}[1]{\lceil}{\rceil}{#1}
\DeclarePairedDelimiterX{\angled}[1]{\langle}{\rangle}{#1}

\newcommand{\NP}{\ensuremath{\textnormal{NP}}}
\newcommand{\coNP}{\ensuremath{\textnormal{coNP}}}
\newcommand{\poly}{\ensuremath{\textnormal{poly}}}

\newcommand{\Ext}{\mathsf{Ext}}

\DeclareMathOperator{\vc}{\ensuremath{\mathsf{vc}}}

\newcommand{\diam}{\operatorname{diam}}
\newcommand{\dist}{\operatorname{dist}}

\EventEditors{Philip Bille, Seth Pettie, and Sabine Storandt}
\EventNoEds{3}
\EventLongTitle{34th Annual European Symposium on Algorithms (ESA 2026)}
\EventShortTitle{ESA 2026}
\EventAcronym{ESA}
\EventYear{2026}
\EventDate{August 31--September 4, 2026}
\EventLocation{L'Aquila, Italy}
\EventLogo{}
\SeriesVolume{388}
\ArticleNo{85}

\begin{document}

\maketitle

\begin{abstract}
\textsc{$H$-Packing} asks whether a graph $G$ contains $k$ vertex-disjoint copies of a fixed pattern graph $H$.
Via the standard reduction to \textsc{$d$-Set Packing}, one obtains generic kernels with $O(k^{|V(H)|-1})$ vertices and $O(k^{|V(H)|})$ edges.
We revisit the question of beating these bounds for specific patterns $H$.

Our main results concern subdivided stars.
Let $S_{d_1,d_2}$ denote the subdivided star with $d_1$ branches of length $1$ and $d_2$ branches of length $2$.
We obtain kernels with $O(k^2)$ vertices and $O(k^3)$ edges for $P_5=S_{0,2}$, for $S_{1,2}$, and for every $S_{d_1,1}$, kernels with $O(k^4)$ vertices and $O(k^6)$ edges for every fixed $S_{d_1,d_2}$ with $d_1\ge 1$, and a kernel with $O(k^2)$ vertices and $O(k^4)$ edges for the paw.
Our proofs proceed in two steps. First, we reduce to instances in which all but a small part of the graph is independent, or in which the graph has a small vertex cover. Second, we reduce the independent side by keeping only a bounded number of witness vertices for each subset of the small part.

On the negative side, we prove a lower bound for the line $S_{0,d}$.
For every $d\ge 3$ and every $\varepsilon>0$, \textsc{$S_{0,d}$-Packing} does not admit a compression of size $O(k^{d-\varepsilon})$ unless $\NP\subseteq \coNP/\poly$.
Thus, deleting a single vertex from the pattern may, surprisingly, make kernelization provably harder, showing that compressibility of \textsc{$H$-Packing} is not monotone under taking induced subgraphs.
\end{abstract}

\section{Introduction}

For a fixed graph $H$, the following \textsc{$H$-Packing} problem~\cite{KirkpatrickH83,Yuster07} is a fundamental problem in theoretical computer science, extensively studied in the literature on parameterized algorithms~\cite{DellM12,FellowsKNRRSTW08,FellowsHRST04,LokshtanovMS18,Moser09}.

\problemdef{$H$-Packing:}{An undirected graph $G$ and an integer $k\in \mathbb{Z}_{\geq 0}$.}{Does $G$ contain $k$ vertex-disjoint (not necessarily induced) subgraphs iso\-\hspace*{5.2em}morphic to $H$?}

One of the lines of research for this problem is \emph{kernelization}, which refers to a polynomial-time algorithm that transforms an instance into a possibly smaller equivalent instance.
For general~$H$, the standard upper bounds come from reducing \textsc{$H$-Packing} to \textsc{$d$-Set Packing}:
\problemdef{$d$-Set Packing:}{Universe $U$, set family $\mathcal F\subseteq \binom{U}{d}$, and integer $k\in \mathbb{Z}_{\geq 0}$.}{Does $\mathcal F$ contain $k$ pairwise disjoint sets?}
Fellows et al.~\cite{FellowsKNRRSTW08} provided a kernel of $O(k^d)$ sets for this problem, while Abu-Khzam~\cite{Abu-KhzamSet10} gave a kernel with $O(k^{d-1})$ elements and $O(k^d)$ sets.
By reducing \textsc{$H$-Packing} to \textsc{$d$-Set Packing} in the standard way and then kernelizing the resulting \textsc{Set Packing} instance, one obtains an upper bound of $O(k^{|V(H)|-1})$ vertices and $O(k^{|V(H)|})$ edges for \textsc{$H$-Packing}.
Indeed, one may lift the \textsc{Set Packing} kernel back to an \textsc{$H$-Packing} instance by keeping the vertices and edges that appear in the surviving copies of $H$.
For connected~$H$, the vertex bound $O(k^{|V(H)|-1})$ was obtained earlier by Moser~\cite{Moser09}.
This leads to the following question:
\begin{quote}
Can one beat the \textsc{$d$-Set Packing} kernel bound for \textsc{$H$-Packing}?
\end{quote}

Known positive answers were confined to a few patterns.
Star packing, that is, \textsc{$K_{1,d}$-Packing}, has been studied in the kernelization literature.
Prieto-Rodriguez and Sloper~\cite{PrietoS06} obtained an $O(k^2)$-vertex kernel for general \textsc{$K_{1,d}$-Packing}, and this was improved to a $O(k)$-vertex kernel \cite{ChenFSWY19,Xiao17}.
The special case $P_3$ has attracted attention: Prieto-Rodriguez and Sloper~\cite{PrietoS06} obtained a $15k$-vertex kernel, and later work improved this to the best bound of $5k$ vertices~\cite{FellowsGMN11,LiYC22,WangNFC10}.
Dell and Marx~\cite{DellM12} showed that when $H$ is a path on $\ell$ vertices, one can beat the general bound with kernels of $O(k^{2.5})$ edges for $\ell=4$ and $O(k^3)$ edges for $\ell\geq 5$.
Cerven{\'{y}}, Choudhary, and Such{\'{y}}~\cite{CervenyCS24} studied the hitting counterpart, obtaining kernels of $O(k^2)$ edges for $\ell=4,5$ and $O(k^4)$ edges for general~$\ell$.\footnote{These kernel bounds for the hitting counterpart do not seem to imply analogous bounds for the packing problem; this was also confirmed to us by Ond\v{r}ej Such\'{y} (private communication). Intuitively, preserving a small hitting set appears substantially easier than preserving many pairwise disjoint copies.}
However, beyond these examples, the list of graphs $H$ known to beat the \textsc{$d$-Set Packing} bound remained very short.

This paper gives several new affirmative answers.
At the same time, it shows that the picture is subtler than expected.
Concretely, we show that deleting a single vertex from the pattern may lead to provably worse kernelization behavior.
This is surprising, and suggests that a more systematic understanding of kernelization for \textsc{$H$-Packing} is still needed.

\longonly{
\paragraph*{Our Results.}
}
\shortonly{
\subparagraph*{Our Results.}
}
Our main focus is the subdivided star $S_{d_1,d_2}$, which has $d_1$ branches of length $1$ and $d_2$ branches of length $2$ (\Cref{fig:subdivided_star}).
In particular, $S_{1,1}=P_4$ ($4$-vertex path), $S_{0,2}=P_5$ ($5$-vertex path), and $S_{d,0}=K_{1,d}$ ($d$-leaf star).
For the case $d_2=0$ (that is, stars), linear-vertex kernels are already known~\cite{ChenFSWY19,Xiao17}.

\begin{figure}[t]
\centering
\begin{tikzpicture}[
    x=0.55cm,y=0.55cm,
    baseline=(current bounding box.center),
    every node/.style={inner sep=0pt, outer sep=0pt},
    line cap=round, line join=round,
    ink/.style={draw=black, line width=1.0pt},
    dot/.style={fill=black, draw=black},
    panelcaption/.style={font=\small, align=center}
]
\def\captionY{-2.5}

\begin{scope}[xshift=0cm]
    \def\yA{2.4}
    \def\yOne{1.2}
    \def\yTwo{0.0}
    \def\r{1.9pt}
    \def\braceGap{0.38}

    \coordinate (A) at (0.35,\yA);
    \fill[dot] (A) circle (\r);

    \foreach \name/\x in {
        P1/-2.10,
        P2/-1.40,
        P3/-0.70,
        P4/0.00,
        P5/0.70,
        P6/1.40,
        P7/2.10,
        P8/2.80
    }{
        \coordinate (\name) at (\x,\yOne);
    }

    \draw[ink] (A) -- (P1);
    \draw[ink] (A) -- (P2);
    \draw[ink] (A) -- (P4);
    \foreach \p in {P1,P2,P4}{\fill[dot] (\p) circle (\r);}
    \node[scale=0.75] at ($(P2)!0.5!(P4)$) {$\cdots$};

    \coordinate (Q1) at ($(P5)+(0,{-\yOne})$);
    \coordinate (Q2) at ($(P6)+(0,{-\yOne})$);
    \coordinate (Q3) at ($(P8)+(0,{-\yOne})$);
    \draw[ink] (A) -- (P5) -- (Q1);
    \draw[ink] (A) -- (P6) -- (Q2);
    \draw[ink] (A) -- (P8) -- (Q3);
    \foreach \p in {P5,P6,P8,Q1,Q2,Q3}{\fill[dot] (\p) circle (\r);}
    \node[scale=0.75] at (2.10,\yOne) {$\cdots$};

    \draw[ink,decorate,decoration={brace,mirror,amplitude=6pt}]
        (-2.42,\yOne-\braceGap) -- (0.18,\yOne-\braceGap)
        node[midway,yshift=-0.50cm] {$d_1$};

    \draw[ink,decorate,decoration={brace,mirror,amplitude=8pt}]
        (0.36,\yTwo-\braceGap) -- (3.08,\yTwo-\braceGap)
        node[midway,yshift=-0.60cm] {$d_2$};

    \node[panelcaption] at (0.35,\captionY) {(a) $S_{d_1,d_2}$};
\end{scope}
\phantomsubcaption\label{fig:subdivided_star}

\begin{scope}[xshift=4.2cm]
    \def\yA{2.4}
    \def\yOne{1.2}
    \def\yTwo{0.0}
    \def\r{1.9pt}
    \def\braceGap{0.38}

    \coordinate (A) at (-0.70,\yA);
    \fill[dot] (A) circle (\r);

    \foreach \name/\x in {
        P1/-2.10,
        P2/-1.40,
        P3/-0.70,
        P4/0.00,
        P5/0.70,
        P6/1.40,
        P7/2.10,
        P8/2.80
    }{
        \coordinate (\name) at (\x,\yOne);
    }

    \draw[ink] (A) -- (P1);
    \draw[ink] (A) -- (P2);
    \draw[ink] (A) -- (P4);
    \foreach \p in {P1,P2,P4}{\fill[dot] (\p) circle (\r);}
    \node[scale=0.75] at ($(P2)!0.5!(P4)$) {$\cdots$};

    \coordinate (Q1) at ($(P5)+(0,{-\yOne})$);
    \draw[ink] (A) -- (P5) -- (Q1);
    \foreach \p in {P5,Q1}{\fill[dot] (\p) circle (\r);}

    \draw[ink,decorate,decoration={brace,mirror,amplitude=6pt}]
        (-2.42,\yOne-\braceGap) -- (0.18,\yOne-\braceGap)
        node[midway,yshift=-0.50cm] {$d_1$};

    \node[panelcaption] at (-0.70,\captionY) {(b) $S_{d_1,1}$};
\end{scope}
\phantomsubcaption\label{fig:Sd1}

\begin{scope}[xshift=7.0cm]
    \def\yA{2.4}
    \def\yOne{1.2}
    \def\yTwo{0.0}
    \def\r{1.9pt}

    \foreach \name/\x in {
        P1/-1.40,
        P2/-0.70,
        P3/0.00,
        P4/0.70,
        P5/1.40
    }{
        \coordinate (\name) at (\x,\yOne);
    }

    \coordinate (A) at (0.00,\yA);
    \fill[dot] (A) circle (\r);

    \draw[ink] (A) -- (P1);
    \fill[dot] (P1) circle (\r);

    \coordinate (Q1) at ($(P3)+(0,{-\yOne})$);
    \coordinate (Q2) at ($(P5)+(0,{-\yOne})$);
    \draw[ink] (A) -- (P3) -- (Q1);
    \draw[ink] (A) -- (P5) -- (Q2);
    \foreach \p in {P3,P5,Q1,Q2}{\fill[dot] (\p) circle (\r);}

    \node[panelcaption] at (0.00,\captionY) {(c) $S_{1,2}$};
\end{scope}
\phantomsubcaption\label{fig:S12}

\begin{scope}[xshift=10.0cm]
    \def\r{1.9pt}
    \def\s{1.5}
    \pgfmathsetmacro{\h}{sqrt(3)/2*\s}

    \coordinate (A) at (0,0);
    \coordinate (B) at (\s,0);
    \coordinate (C) at ({\s/2},\h);
    \coordinate (D) at ({\s/2},{\h+\s});

    \draw[ink] (A) -- (B) -- (C) -- cycle;
    \draw[ink] (C) -- (D);

    \foreach \p in {A,B,C,D}{\fill[dot] (\p) circle (\r);}

    \node[panelcaption] at ({\s/2},\captionY) {(d) Paw};
\end{scope}
\phantomsubcaption\label{fig:paw}
\end{tikzpicture}
\caption{Pattern graphs used throughout the paper.}
\end{figure}
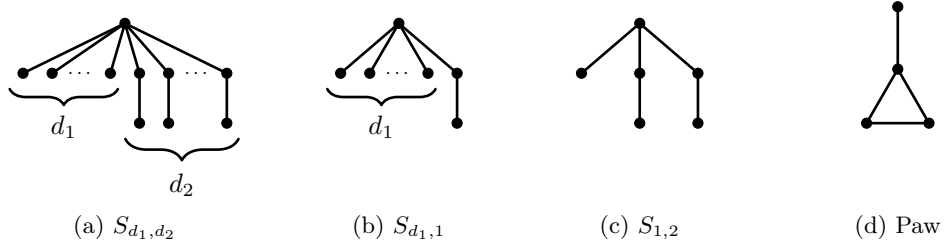

Our first theorem extends this line beyond stars, giving quadratic-vertex and cubic-edge kernels when $d_2=1$ and in the two additional cases $(d_1,d_2)\in\{(0,2),(1,2)\}$.

\begin{restatable}{theorem}{ShortSubdividedStarKernel}\label{thm:short-subdivided-stars}
For every fixed subdivided star $H=S_{d_1,d_2}$ with $d_2 = 1$ or $(d_1,d_2)\in\{(0,2),(1,2)\}$, \textsc{$H$-Packing} admits a kernel with $O(k^2)$ vertices and $O(k^3)$ edges.
\end{restatable}
For the family $H=S_{d_1,1}$, the hidden constants depend polynomially on $d_1$.

At a high level, the argument has two steps, both using reductions based on the expansion lemma: first, one finds a small set that captures the non-trivial part of the graph, and then one reduces the remaining independent side.

The cleanest case is $P_5=S_{0,2}$.
The key step is to obtain a vertex cover of size $O(k)$.
Let $C$ be the union of vertex sets of a maximal family of vertex-disjoint copies of $P_5$; then $|C|=O(k)$, and $C$ meets every copy of $P_5$ by maximality.
A reduction based on the expansion lemma reduces the number of non-trivial components of $G-C$ to $O(k)$.
We then obtain a linear-size vertex cover, because every connected $P_5$-free graph has a constant-size vertex cover.
Thus each non-trivial component of $G-C$ contributes only constantly many cover vertices, while all remaining vertices are isolated in $G-C$.
Consequently the resulting graph has a vertex cover of size~$O(k)$.

The second step is to reduce the remaining independent side by keeping only a bounded number of vertices for each relevant subset of the vertex cover.
For $P_5$, applying this to the linear-size vertex cover obtained in the first step yields a quadratic-vertex kernel.

For $S_{1,2}$ and $S_{d_1,1}$, the same first idea starts the proof, but the vertex-cover argument used for $P_5$ is too rigid.
To handle this, we introduce the notion of an \emph{isolating set}.
Concretely, we seek an isolating set $C$ for which $V(G)$ admits a partition $V(G)=C\uplus Z\uplus I$, where $I$ is independent and there is no edge between $I$ and $Z$.
When $Z=\emptyset$, this set $C$ is exactly a vertex cover.
Here one again applies reduction based on the expansion lemma to obtain $|C|=O(k)$, but now keeps an additional set $Z$ of size $O(k^2)$ that contains the non-trivial surviving part, leaving the rest in the independent set $I$.
The second step then reduces this $I$-side exactly as in the $P_5$ case.
This yields \Cref{thm:short-subdivided-stars}.

We next turn to arbitrary subdivided stars with $d_1 \ge 1$.

\begin{restatable}{theorem}{SubdividedStarKernel}\label{thm:general-result}
For $d_1 \ge 1$, \textsc{$S_{d_1,d_2}$-Packing} admits a kernel with $O(k^4)$ vertices and $O(k^6)$~edges.
\end{restatable}

The kernelization landscape for subdivided stars is summarized in \Cref{fig:subdivided-star-landscape}.

\begin{figure}[t]
\centering
\begin{tikzpicture}[
    x=0.62cm,y=0.62cm,
    every node/.style={font=\small},
    axis/.style={->, very thick, >=stealth},
    grid/.style={gray!55, line width=0.3pt},
    linearfill/.style={fill=gray!6},
    shortfill/.style={fill=gray!15},
    generalfill/.style={fill=gray!38},
    hardfill/.style={fill=gray!62},
    regionlabel/.style={
        font=\scriptsize,
        align=center,
        fill=white,
        fill opacity=0.9,
        text opacity=1,
        inner sep=1.5pt,
        rounded corners=1pt
    }
]
\pgfmathsetmacro{\xmax}{6}
\pgfmathsetmacro{\ymax}{6}
\pgfmathsetmacro{\xmaxd}{\xmax - 1}
\pgfmathsetmacro{\ymaxd}{\ymax - 1}

\foreach \x in {0,...,\xmaxd} {
    \fill[shortfill] (\x,0) rectangle (\x + 1,1);
}
\foreach \x/\y in {0/1,1/1} {
    \fill[shortfill] (\x,\y) rectangle (\x + 1,\y + 1);
}
\fill[linearfill] (0,0) rectangle (1,1);

\foreach \x in {2,...,\xmaxd} {
    \fill[generalfill] (\x,1) rectangle (\x + 1,2);
}
\foreach \x in {1,...,\xmaxd} {
    \foreach \y in {2,...,\ymaxd} {
        \fill[generalfill] (\x,\y) rectangle (\x + 1,\y + 1);
    }
}

\foreach \y in {2,...,\ymaxd} {
    \fill[hardfill] (0,\y) rectangle (1,\y + 1);
}

\draw[gray!45, very thick] (0,0) -- (6,0) -- (6,1) -- (2,1) -- (2,2) -- (0,2) -- cycle;
\draw[gray!65, very thick] (2,1) -- (6,1) -- (6,6) -- (1,6) -- (1,2) -- (2,2) -- cycle;
\draw[gray!85, very thick] (0,2) -- (1,2) -- (1,6) -- (0,6) -- cycle;

\foreach \i in {0,...,\xmaxd} {
    \node at (\i + .5, -.45) {\i};
}
\foreach \i in {0,...,\ymaxd} {
    \node at (-.45, \i + .5) {\number\numexpr\i+1\relax};
}
\foreach \i in {0,...,\xmax} {
    \draw[grid] (\i,0) -- (\i,\ymax);
}
\foreach \i in {0,...,\ymax} {
    \draw[grid] (0,\i) -- (\xmax,\i);
}
\draw[axis] (0,0) -- (\xmax + .35,0) node[right] {$d_1$};
\draw[axis] (0,0) -- (0,\ymax + .35) node[above] {$d_2$};

\node[font=\scriptsize] at (0.5,0.5) {$O(k)$};
\node[align=center] at (3.9,0.55) {$O(k^2)$};
\node[align=center] at (3.7,3.9) {$O(k^4)$};
\node[align=center,rotate=90] at (0.5,4.0) {no $O(k^{d-\epsilon})$};
\end{tikzpicture}
\caption{
Kernelization landscape for subdivided stars.
The lightest cell, corresponding to $(d_1,d_2)=(0,1)$, admits an $O(k)$-vertex kernel.
The light-gray region consists of the remaining cases with $O(k^2)$-vertex kernels.
The darker gray region consists of the additional cases with $d_1 \ge 1$ that admit $O(k^4)$-vertex kernels.
The darkest region is the line $d_1=0$, $d_2\ge 3$, for which we rule out compressions of size $O(k^{d-\epsilon})$ for every $\epsilon>0$, unless $\NP\subseteq \coNP/\poly$.
}
\label{fig:subdivided-star-landscape}
\end{figure}
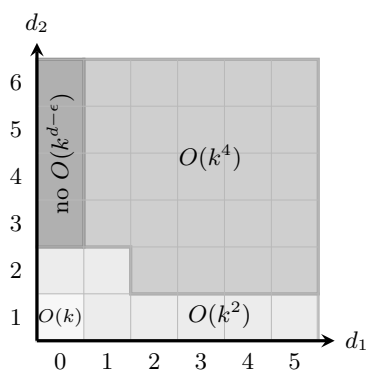

The overall strategy again has two steps.
First, we reduce the instance to one with a vertex cover of size $O(k^2)$.
Unlike the cases covered by \Cref{thm:short-subdivided-stars}, this vertex cover is no longer linear in $k$.
Starting from a maximal packing, we take a maximal matching in the remaining graph.
Because the vertices outside this matching form an independent set, the packing vertices together with the matched vertices already form a vertex cover.
The key point is then to show, by a degree analysis around the matching, that after deleting vertices that belong to no copy of $H$, only $O(k^2)$ matched vertices need to be kept.
Second, we reduce the remaining part outside this vertex cover using the same general approach as in the proof of \Cref{thm:short-subdivided-stars}.

The assumption $d_1\ge 1$ in \Cref{thm:general-result} is essential for the second step.
We introduce the notion of the pattern's \emph{$r$-replaceability}, defined formally in \Cref{sec:bounded-vc-separated}.
Roughly speaking, this notion controls how efficiently the part outside the vertex cover or isolating set can be represented once a partition $V(G)=C\uplus Z\uplus I$ is given.
Our general reduction theorem shows that, from such a partition, one obtains a kernel with $O(|C|^r+|Z|)$ vertices (\Cref{thm:replaceable_kernel}).

For subdivided stars, the distinction between the cases $d_1=0$ and $d_1\ge 1$ turns out to be fundamental.
If $d_1\ge 1$, then the presence of a branch of length $1$ makes it possible to replace the part of a copy outside the cover while preserving the pattern, and $S_{d_1,d_2}$ is already $2$-replaceable, which is what ultimately yields the polynomial kernel of \Cref{thm:general-result}.
In contrast, when $d_1=0$, our framework gives only $d_2$-replaceability for $S_{0,d_2}$.
One might suspect that this is merely a limitation of our method.
The next theorem, however, shows that it reflects a genuine structural barrier:

\begin{restatable}{theorem}{SubdividedStarHardness}\label{thm:S0d-lb-result}
For any $d\geq 3$ and $\epsilon>0$, \textsc{$S_{0,d}$-Packing} does not admit a compression with size $O(k^{d-\epsilon})$ unless $\NP\subseteq \coNP/\poly$.
\end{restatable}

\shortonly{The lower bound is by a parameter-preserving reduction from \textsc{Exact $d$-Set Packing}, which by Dell and Marx has no compression of size $O(k^{d-\epsilon})$ unless $\NP\subseteq \coNP/\poly$. For each element $u$, the construction creates an edge $c_ui_u$, and for each set $F$ it creates a vertex $x_F$ adjacent precisely to the vertices $c_u$ with $u\in F$. Hence each selected set yields a copy of $S_{0,d}$ centered at $x_F$, whose $d$ length-two branches encode the elements of $F$. Conversely, since there are exactly $dk$ vertices of the form $c_u$, any packing of $k$ such stars forces the centers to be set-vertices $x_F$ and recovers $k$ pairwise disjoint $d$-sets. Thus the parameter stays $k$ and the exponent $d$ is inherited unchanged. See the full version for the full construction and proof.}

This shows that deleting a single vertex from the pattern can make kernelization strictly harder: \textsc{$S_{1,7}$-Packing} admits a kernel with $O(k^6)$ edges by \Cref{thm:general-result}, whereas \textsc{$S_{0,7}$-Packing} does not admit a compression of size $O(k^{7-\epsilon})$ unless $\NP\subseteq \coNP/\poly$.
It also highlights the role of replaceability: $S_{1,7}$ is already $2$-replaceable, while $S_{0,7}$ is only $7$-replaceable, and this gap is reflected in their kernelization complexity.

Finally, we consider the paw (\Cref{fig:paw}), i.e., a triangle with a pendant edge.

\begin{restatable}{theorem}{PawKernel}\label{thm:paw-result}
\textsc{Paw-Packing} admits a kernel with $O(k^2)$ vertices and hence $O(k^4)$ edges.
\end{restatable}

This is the first example, to the best of our knowledge, of beating the general \textsc{$d$-Set Packing} bound for a pattern that contains a cycle.
The proof again follows a two-step kernelization scheme, but now in a form adapted to triangles.
Starting from a maximal paw packing, one reduces the number of triangle components outside the initial set and controls large matchings around its vertices.
This yields an equivalent instance with a vertex cover $C$ of size $O(k^2)$.
The remaining outside vertices matter only through $O(k^2)$ edges of $C$ that can participate in triangles, and a second matching-based kernelization using these edges gives the kernel in \Cref{thm:paw-result}.

\longonly{
\paragraph*{Related Work.}
}
\shortonly{
\subparagraph*{Related Work.}
}

The algorithmic study of \textsc{$H$-Packing} goes back to Hell and Kirkpatrick's work on generalized matching problems~\cite{KirkpatrickH78,KirkpatrickH83}: for fixed $H$, the problem is polynomial-time solvable when each connected component of $H$ has at most two vertices and NP-hard when $H$ has at least three vertices.
For broader background on the algorithmic and computational aspects of graph packing, we refer to Yuster's survey~\cite{Yuster07}.

From the parameterized perspective, the generic benchmark is \textsc{$d$-Set Packing}.
For that problem, the fastest known generic randomized algorithms are due to Bj{\"o}rklund, Husfeldt, Kaski, and Koivisto~\cite{BjorklundHKK17}; in particular, they solve \textsc{$3$-Set Packing} in time $O^*(3.328^{k})$ and \textsc{$4$-Set Packing} in time $O^*(7.312^k)$.\footnote{The $O^*$ notation suppresses polynomial factors.}
Deterministically, Nederlof~\cite{Nederlof25} very recently obtained an $O^*(2^{dk+o(dk)})$-time algorithm for \textsc{$d$-Set Packing}.
Via the standard reduction from \textsc{$H$-Packing} to \textsc{$d$-Set Packing}, this yields fixed-parameter algorithms for \textsc{$H$-Packing}.
For the special case $P_3$, Zehavi~\cite{Zehavi15} gave a deterministic algorithm with running time $O^*(6.75^k)$ based on Feng et al.~\cite{FengWC14}.

Concerning kernelization, generic kernels for \textsc{$H$-Packing} with $O\!\left(k^{|V(H)|-1}\right)$ vertices and $O(k^{|V(H)|})$ edges follow from \textsc{$d$-Set Packing}~\cite{Abu-KhzamSet10,FellowsKNRRSTW08}.
For connected $H$, Moser~\cite{Moser09} had earlier obtained the vertex bound $O\!\left(k^{|V(H)|-1}\right)$ directly for \textsc{$H$-Packing}.
On the lower-bound side, Dell and Marx~\cite{DellM12} showed that \textsc{$H$-Packing} does not admit a kernel of $O(k^{2-\epsilon})$ edges for any $\epsilon>0$ and any connected graph $H$ with at least three vertices unless $\NP \subseteq \coNP/\poly$.
For cliques, they further proved that \textsc{$K_d$-Packing} does not admit a kernel of $O(k^{d-1-\epsilon})$ edges for any $\epsilon>0$ and any $d\ge 3$ unless $\coNP \subseteq \NP/\poly$.

Before the present paper, the only graphs known to beat these generic upper bounds were essentially stars and paths.
For stars, the kernel size was improved from quadratic to linear~\cite{ChenFSWY19,PrietoS06,Xiao17}; for the special case $P_3$, the kernel size was later reduced further to $5k$ vertices~\cite{FellowsGMN11,LiYC22,PrietoS06,WangNFC10}.
For paths, Dell and Marx~\cite{DellM12} obtained kernels with $O(k^{2.5})$ edges for $P_4$ and $O(k^3)$ edges for $P_\ell$ with $\ell\geq 5$.


A variant that considers packing induced subgraphs instead of subgraphs, \textsc{Induced $H$-Packing}, has also been studied.
Via the standard reduction to \textsc{Set Packing}, one obtains a compression with $O(k^{|V(H)|-1})$ elements and $O(k^{|V(H)|})$ sets.
Fomin et al.~\cite{FominLLSTZ19} obtained an $O(k^{5/3})$-vertex kernel for induced $P_3$-Packing via expansion lemma, and Bessy et al.~\cite{BessyBTW23} improved this to $O(k)$ vertices via rainbow matching.
These papers also obtained improved kernels for packing oriented triangles in tournaments.

\longonly{
\paragraph*{Organization.}
\Cref{sec:preliminaries} introduces notation and the expansion lemma.
The next three sections are devoted to subdivided stars.
\Cref{sec:direct-reductions} handles the special cases $P_5$, $S_{1,2}$, and $S_{d_1,1}$ by reducing them to instances with linear-size isolating sets, while \Cref{sec:first-stage-vc} handles general subdivided stars by reducing them to instances with a quadratic-size vertex cover.
\Cref{sec:bounded-vc-separated} then develops the general reduction from isolating sets and bounded vertex cover, establishes the replaceability properties, and completes the proofs of \Cref{thm:short-subdivided-stars,thm:general-result}.
\Cref{sec:lowerbound} proves \Cref{thm:S0d-lb-result}.
Finally, \Cref{sec:paw} proves \Cref{thm:paw-result}.
}
\shortonly{
\subparagraph*{Organization.}
\Cref{sec:preliminaries} introduces notation and the expansion lemma.
The next three sections are devoted to subdivided stars.
\Cref{sec:direct-reductions} reduces the special cases $P_5$, $S_{1,2}$, and $S_{d_1,1}$ to instances with small isolating sets.
\Cref{sec:bounded-vc-separated} then develops the general reduction from isolating sets and bounded vertex cover, establishes the replaceability properties, and completes the proofs of \Cref{thm:short-subdivided-stars,thm:general-result}.
Proofs of statements marked with \proofomittedmark{} are omitted, and \Cref{thm:S0d-lb-result,thm:paw-result} are proved only in the full version.
}

\section{Preliminaries}\label{sec:preliminaries}

We assume that the reader is familiar with the basic notions of parameterized complexity; for background, we refer to Cygan et al.~\cite{Cygan2015}.

\subparagraph*{Graphs.}
In this paper, all graphs are finite, simple, and undirected.
For a graph $G$ and a vertex set $X\subseteq V(G)$, we write $G[X]$ for the induced subgraph on $X$ and $N_G(X)$ for the neighborhood of $X$ in $G$; if $X=\{x\}$, we abbreviate $N_G(\{x\})$ to $N_G(x)$.
For two disjoint vertex sets $X,Y\subseteq V(G)$, we write $G[X,Y]$ for the bipartite graph obtained from $G[X\cup Y]$ by removing all edges in $G[X]$ and $G[Y]$.
We denote a path by the sequence of its vertices, so $(v_1,\dots,v_r)$ is the path with edges $v_iv_{i+1}$ for $i\in\{1,\dots,r-1\}$.
For $D\subseteq V(G)$ and $X\subseteq V(G)$, let $N_D(X):=N_G(X)\cap D$, and for an integer $r\ge 1$, let $N_D^{\le r}(X)$ denote the set of vertices $v\in D$ such that, for some $x\in X$, there is a path $(x,\dots,v)$ of length at most $r$, or equivalently, at most $r$ edges, whose internal vertices all lie in $D$.

We write $\deg_G(v)$ for the degree of a vertex $v$ in $G$, $\Delta(G)$ for the maximum degree of~$G$, $\dist_G(u,v)$ for the distance between $u$ and $v$ in~$G$, and $\diam(G)$ for the diameter of a connected graph $G$.
In particular, every subdivided star $S_{d_1,d_2}$ has diameter at most $4$.

\subparagraph*{Kernelization.}
Two instances $(G,k)$ and $(G',k')$ of \textsc{$H$-Packing} are \emph{equivalent} if either both are yes-instances or both are no-instances.
A \emph{kernelization} of size $f(k)$ for a parameterized problem is a polynomial-time algorithm that maps each instance $(x,k)$ to an equivalent instance $(x',k')$ of the same problem whose size is bounded by a computable function $f(k)$.
More generally, a \emph{compression} of size $f(k)$ outputs an equivalent instance of some parameterized problem that can be encoded using $O(f(k))$ bits.
In this paper, for graph problems we measure the size of a reduced instance in terms of the number of vertices and edges of the reduced graph.
A reduction rule that replaces an instance $(G,k)$ by an instance $(G',k')$ is \emph{safe} if $(G,k)$ and $(G',k')$ are equivalent.

For a fixed graph $H$, we will repeatedly use the following elementary reduction rule.

\begin{rrule}\label{rule:irrelevant-element}
If a vertex or an edge of $G$ is contained in no copy of $H$, delete it.
\end{rrule}

\markproofomitted
\begin{lemma}\label{lem:irrelevant-element-safe}
\Cref{rule:irrelevant-element} is safe.
\end{lemma}
\longonly{
\begin{proof}
Let $G'$ be obtained from $G$ by deleting a vertex or an edge that is contained in no copy of $H$.
Every $H$-packing in $G'$ is also an $H$-packing in $G$.
Conversely, no copy of $H$ in $G$ uses the deleted vertex or edge, so every $H$-packing in $G$ is also contained in $G'$.
Hence $(G,k)$ and $(G',k)$ are equivalent.
\end{proof}
}

After exhausting the rule, we may assume that every vertex and every edge of $G$ belongs to a copy of $H$.
Since $H$ is fixed, applicability of the rule can be checked in polynomial time.

\subparagraph{Expansion lemma.}
The expansion lemma, originally introduced by Thomass{\'e}~\cite{Thomasse10} in his quadratic-kernel algorithm for \textsc{Feedback Vertex Set}, is a standard tool in kernelization.
For background and applications, we refer to the textbook on kernelization by Fomin et al.~\cite{FominLSZ19}.
In this paper, we repeatedly use the following formulation due to Fomin et al.~\cite{FominLLSTZ19}.
\begin{lemma}[Expansion lemma~\cite{FominLLSTZ19}]\label{lem:expansion}
Let $B=(L\uplus R,E)$ be a bipartite graph and $q\ge 1$ be an integer.
There is a polynomial-time algorithm that computes sets $\widehat L\subseteq L$, $\widehat R\subseteq R$, and an edge set $M\subseteq E$ such that
\begin{enumerate}
    \item $|\widehat R|\le q|\widehat L|$,
    \item every vertex of $L\setminus \widehat L$ is incident with exactly $q$ edges of $M$,
    \item every vertex of $R\setminus \widehat R$ is incident with at most one edge of $M$, and
    \item there is no edge of $B$ between $\widehat L$ and $R\setminus \widehat R$.
\end{enumerate}
\end{lemma}



\longonly{
\section{Reducing the Special Cases to Instances with a Linear-Size Isolating Set}\label{sec:direct-reductions}
}
\shortonly{
\section{Reducing to Small Structured Instances}\label{sec:direct-reductions}
}

\longonly{
As outlined in the introduction, our goal in this section is to reduce the special cases to instances equipped with small isolating sets, which we now define formally.
}
\shortonly{
As outlined in the introduction, this section contains the first-stage reductions for subdivided stars.
We first treat the special cases $P_5$, $S_{1,2}$, and $S_{d_1,1}$.
We use the following notion.
}

\begin{definition}[Isolating set]
Let $G$ be a graph.
A set $C\subseteq V(G)$ is an \emph{isolating set} of $G$ if there is a partition
\(
V(G)=C\uplus Z\uplus I
\)
such that $I$ is an independent set and $E(Z,I)=\emptyset$.
We say that such a partition \emph{witnesses} that $C$ is an isolating set.
In particular, if $Z=\emptyset$, then $C$ is a vertex cover of $G$.
\end{definition}

\longonly{
More precisely, we prove that for each fixed graph $H\in \{S_{0,2}=P_5,S_{1,2},S_{d_1,1}\}$, every instance of \textsc{$H$-Packing} can be reduced in polynomial time to an equivalent instance equipped with an isolating set $C$ and a partition $V(G)=C\uplus Z\uplus I$ witnessing that $C$ is an isolating set, where $|C|=O(k)$ and $|Z|=O(k^2)$.
}
\shortonly{
More precisely, for each fixed graph $H\in \{S_{0,2}=P_5,S_{1,2},S_{d_1,1}\}$, we reduce \textsc{$H$-Packing} to an equivalent instance equipped with an isolating set $C$ and a witnessing partition $V(G)=C\uplus Z\uplus I$, where $|C|=O(k)$ and $|Z|=O(k^2)$.
For each fixed $S_{d_1,d_2}$, we also reduce to an equivalent instance with a vertex cover of size $O(d^5k^2)$, where $d=d_1+d_2$.
}

\longonly{
The strategy is the same in all three cases.
}
\shortonly{
We first handle the three special cases; the strategy is the same for them.
}
Fix $H=S_{d_1,d_2}$ among these three graphs and write $d:=d_1+d_2$.
We begin by computing a maximal family $\mathcal{P}$ of pairwise vertex-disjoint copies of $H$ and setting
\(
C:=\bigcup_{P\in\mathcal P}V(P).
\)
This can be done in polynomial time.
If $|\mathcal P|\ge k$, then $(G,k)$ is a yes-instance.
Hence we may assume that $|\mathcal P|<k$.
It follows that $|C|=|V(H)|\cdot |\mathcal P|<|V(H)|k\le (2d+1)k$, and maximality implies that $G-C$ is $H$-free.

\longonly{
Next, we use the expansion lemma to reduce the number of non-trivial connected components of $G-C$ (\Cref{subsec:common-component-reduction}).
The remaining work is case-specific: for each of the three graphs~$H$, we show that the surviving part of $G-C$ can be partitioned into a set $Z$ of size $O(k^2)$ and an independent set $I$.
This gives the desired partition $V(G)=C\uplus Z\uplus I$ and shows that $C$ is an isolating set.
}
\shortonly{
For these special cases, we next use the expansion lemma to reduce the number of non-trivial connected components of $G-C$ (\Cref{subsec:common-component-reduction}).
The remaining work is case-specific: for each graph~$H$, we partition the surviving part of $G-C$ into a set $Z$ of size $O(k^2)$ and an independent set $I$.
This gives the desired isolating-set partition.
The general vertex-cover reduction appears in \Cref{sec:first-stage-vc}.
}


\subsection{Non-trivial component reduction}\label{subsec:common-component-reduction}


Let $\mathcal R$ be the set of connected components of $G-C$ with at least two vertices, and let $J$ be the set of isolated vertices of $G-C$.
Then
\(
V(G)\setminus C = J \uplus \bigl(\biguplus_{R\in\mathcal R}V(R)\bigr).
\)

Define a graph $\Gamma:=(C\uplus \mathcal R,F)$ by joining $c\in C$ and $R\in\mathcal R$ whenever $c$ has a neighbor in $R$ in the graph $G$.
Applying \Cref{lem:expansion} to $\Gamma$ with $q=d$, we obtain subsets $\widehat C\subseteq C$ and $\widehat{\mathcal R}\subseteq \mathcal R$, define $A:=C\setminus \widehat C$ and $\mathcal D:=\mathcal R\setminus \widehat{\mathcal R}$, and obtain an edge set $M\subseteq F$ such that
\begin{enumerate}
    \item $|\widehat{\mathcal R}|\le d|\widehat C|$,
    \item every vertex of $A $ is incident in $M$ with exactly $d$ edges,
    \item every vertex of $\mathcal D$ is incident in $M$ with at most one edge,
    \item there is no edge of $\Gamma$ between $\widehat C$ and $\mathcal D$.
\end{enumerate}
We say that a component $R\in\mathcal D$ is \emph{assigned} to a vertex $c\in A$ if $cR\in M$.
Then the second property implies that exactly $d$ components of $\mathcal D$ are assigned to each vertex of $A$, while the third property implies that each component of $\mathcal D$ is assigned to at most one vertex of $A$.


We apply the following reduction rule.
\begin{rrule}\label{rule:component-reduction}
Delete $A\uplus \bigcup_{R\in\mathcal D}V(R)$ and decrease $k$ by $|A|$.
\end{rrule}

\begin{lemma}\label{lem:component-reduction-safe}
\Cref{rule:component-reduction} is safe.
\end{lemma}

\begin{proof}
Define $G':=G\Bigl[\widehat C\cup J\cup \bigcup_{R\in\widehat{\mathcal R}}V(R)\Bigr]$ and $k':=k-|A|$.
Let $D:=A\cup \bigcup_{R\in \mathcal D}V(R)$ be the set of deleted vertices.
Note that $V(G')=V(G)\setminus D$.

Suppose first that $(G,k)$ is a yes-instance and let $\mathcal Q$ be a packing of $k$ copies of $H$ in $G$.
We claim that every copy in $\mathcal Q$ intersecting $D$ also intersects $A$.
Indeed, this is immediate if the copy uses a vertex of $A$.
Otherwise it uses a vertex of some component $R\in\mathcal D$.
Since $G-C$ is $H$-free, the copy cannot be contained entirely in $G-C$, and therefore it contains a vertex of $C$ adjacent to $R$.
By the fourth property, no vertex of $\widehat C$ is adjacent to a component of $\mathcal D$, so this vertex of $C$ must lie in $A$.
Hence every copy of $\mathcal Q$ intersecting $D$ contains a distinct vertex of $A$.
As the packing is vertex-disjoint, at most $|A|$ copies of $\mathcal Q$ intersect $D$.
Removing these copies leaves at least $k-|A|=k'$ copies disjoint from $D$, and therefore contained in $G'$.
Thus $(G',k')$ is a yes-instance.

Conversely, suppose that $(G',k')$ is a yes-instance, and let $\mathcal Q'$ be a packing of $k'$ copies of $H$ in $G'$.
For each vertex $c\in A$, consider the $d=d_1+d_2$ components of $\mathcal D$ assigned to~$c$.
Choose $d_1$ of them for the branches of length $1$, and in each such component choose a vertex adjacent to $c$.
In each of the remaining $d_2$ components, choose a vertex $x$ adjacent to $c$ and a neighbor $y$ of $x$ in the same component.
These choices yield a copy of $H=S_{d_1,d_2}$ centered at $c$.
Since no component of $\mathcal D$ is assigned to two different vertices of $A$, the resulting copies are pairwise vertex-disjoint and lie in $G[D]$.
These copies are disjoint from $\mathcal Q'$ because $V(G')\cap D=\emptyset$.
Together with $\mathcal Q'$, they yield $k'+|A|=k$ pairwise vertex-disjoint copies of $H$ in $G$.
Thus $(G,k)$ is a yes-instance.
\end{proof}

Let $(G',k')$ be the instance obtained by applying \Cref{rule:component-reduction}.
Replace the current instance by $(G',k')$, and rename $\widehat C$ and $\widehat{\mathcal R}$ as $C$ and $\mathcal R$.
Then $|\mathcal R|\le d|C|<d|V(H)|k$.
Moreover, if a component $R\in\mathcal R$ has no neighbor in $C$, then no copy of $H$ can use a vertex of $R$: indeed, $R\subseteq G-C$ is $H$-free and has no edge to the rest of the graph outside $R$.
Hence all such components may be deleted.

\subsection{Reduction for $P_5$}\label{subsec:direct-p5}

For $H=P_5$, \Cref{rule:component-reduction} already leaves only $O(k)$ non-trivial components outside~$C$.
The key point is that each such component has a vertex cover of constant size.

\markproofomitted
\begin{lemma}\label{lem:p5-component-vc}
Every connected $P_5$-free graph has a vertex cover of size at most $3$.
\end{lemma}
\longonly{
\begin{proof}
Let $F$ be a connected $P_5$-free graph with at least two vertices.

If $F$ is a tree, then every simple path of $F$ has at most four vertices.
Thus $F$ has diameter at most $3$, and therefore $F$ is a star or a double-star.
In either case, $F$ has a vertex cover of size at most $2$.

Assume next that $F$ contains a cycle.
Since every cycle of length at least $5$ contains a $P_5$, every cycle of $F$ has length $3$ or $4$.

Suppose first that $F$ contains a $4$-cycle $C$.
If some vertex $x\notin V(C)$ had a neighbor on $C$, then $x$ together with the four vertices of $C$ would contain a $P_5$.
As $F$ is connected, $V(F)=V(C)$, and therefore $\vc(F)\le 3$.

We may therefore assume that $F$ contains a triangle $T = \{a,b,c\}$.
Every vertex outside $T$ has a neighbor in $T$; otherwise, a shortest path from that vertex to $T$ together with two edges of the triangle would contain a $P_5$.
If $xy\in E(F)$ with $x,y\notin T$, choose a neighbor $a\in T$ of $x$.
Then $(y,x,a,b,c)$ is a $P_5$, a contradiction.
Thus every edge of $F$ has at least one endpoint in $T$, so $T$ is a vertex cover of size $3$.
\end{proof}
}

\begin{lemma}\label{lem:direct-p5-vc}
Given an instance $(G,k)$ of $P_5$-Packing, one can compute in polynomial time an equivalent instance $G'$ with $\vc(G')=O(k)$.
\end{lemma}
\begin{proof}
Apply \Cref{rule:component-reduction} with $H=P_5=S_{0,2}$, and let $(G,k)$ denote the resulting equivalent instance.
Then $|\mathcal R|=O(k)$, and every component of $\mathcal R$ has a neighbor in $C$.

For each component $R\in \mathcal R$, choose by \Cref{lem:p5-component-vc} a vertex cover $X_R$ of $R$ with $|X_R|\le 3$, and set $C^+:=C\cup \bigcup_{R\in\mathcal R}X_R$.
Since $|\mathcal R|=O(k)$, we have $|C^+|=O(k)$.

We claim $C^+$ is a vertex cover of $G$.
Let $uv\in E(G)$.
If $u\in C$ or $v\in C$, then $uv$ is covered by $C\subseteq C^+$.
Otherwise, both endpoints lie in $V(G)\setminus C$ and belong to the same connected component of $G-C$.
This component cannot be trivial, since vertices of $J$ are isolated in $G-C$.
Therefore both endpoints lie in some component $R\in\mathcal R$, so $uv$ is covered by $X_R\subseteq C^+$.
Thus every edge of $G$ has an endpoint in $C^+$, so $\vc(G)\le |C^+|=O(k)$.
\end{proof}

\subsection{Reduction for $S_{1,2}$}\label{subsec:direct-s12}

For $H=S_{1,2}$, we again seek an equivalent instance together with an isolating set $C$ and a partition $V(G)=C\uplus Z\uplus I$ witnessing this, where $|C|=O(k)$ and $|Z|=O(k^2)$.

\markproofomitted
\begin{lemma}\label{lem:s12-structure}
Let $R$ be a connected $S_{1,2}$-free graph.
Then $R$ contains at most six vertices of degree at least $3$.
\end{lemma}
\longonly{
\begin{proof}
Assume that $|V(R)|\ge 6$, and let $P=(v_1,v_2,\dots,v_t)$ be a longest path in $R$.

We first claim that if a vertex $x\notin V(P)$ has a neighbor on $P$, then that neighbor is $v_2$ or $v_{t-1}$.
Indeed, $x$ cannot be adjacent to $v_1$ or $v_t$, since then $P$ would not be longest.
If $x$ were adjacent to some $v_i$ with $3\le i\le t-2$, then $v_i$ together with $x$, $v_{i-1}$, $v_{i-2}$, $v_{i+1}$, and $v_{i+2}$ would form a copy of $S_{1,2}$, a contradiction.

Next we claim that no vertex is at distance $2$ from $P$.
Suppose that $x\notin V(P)$ is at distance $2$ from $P$, and let $x-y-v_i$ be a path from $x$ to $P$.
By the previous paragraph, $y$ is adjacent to either $v_2$ or $v_{t-1}$; by symmetry assume that $y$ is adjacent to $v_2$.
Then $x,y,v_2,\dots,v_t$ is longer than $P$, a contradiction.

We next show that every vertex outside $P$ has degree at most $2$.
Let $x\notin V(P)$.
Since no vertex is at distance $2$ from $P$, the vertex $x$ has a neighbor on $P$, and by the first claim every such neighbor belongs to $\{v_2,v_{t-1}\}$.
Then, by the same argument as the previous paragraph, $x$ has no neighbor outside $P$, so $\deg_R(x)\le 2$.
Consequently every vertex of degree at least $3$ lies on $P$.

Now let $i$ satisfy $4\le i\le t-3$.
We claim that $v_i$ has degree $2$.
It has no neighbor outside $P$, by the first claim.
Suppose that it has a neighbor $v_j$ on $P$ with $j\notin\{i-1,i+1\}$.
If $j=i+2$, then $v_i$ is the center of a copy of $S_{1,2}$ with branches $v_i-v_{i-1}-v_{i-2}$, $v_i-v_{i+2}-v_{i+3}$, and $v_i-v_{i+1}$.
If $j\ge i+3$, then $v_i$ is the center of a copy of $S_{1,2}$ with branches $v_i-v_{i-1}-v_{i-2}$, $v_i-v_j-v_{j-1}$, and $v_i-v_{i+1}$.
The case $j\le i-2$ is symmetric.
This contradiction proves that $\deg_R(v_i)=2$.

Therefore every vertex of degree at least $3$ belongs to $\{v_1,v_2,v_3,v_{t-2},v_{t-1},v_t\}$, so there are at most six such vertices.
\end{proof}
}

We now return to the instance produced by \Cref{rule:component-reduction} with $H=S_{1,2}$.
If that reduction certifies a yes-instance, we return the disjoint union of $k$ copies of $S_{1,2}$.
Otherwise, we have $|C|<6k$ and $|\mathcal R|<18k$.

For each component $R\in \mathcal R$, move every vertex of degree at least $3$ in $R$ to $C$.
By \Cref{lem:s12-structure}, this adds at most $6|\mathcal R|=O(k)$ vertices, so still $|C|=O(k)$.
Let $D$ be the union of all non-trivial components of $G-C$.
Then $\Delta(D)\le 2$.
Moreover, by \Cref{rule:irrelevant-element}, every vertex of $D$ belongs to some copy of $S_{1,2}$.
Since each such copy intersects $C$ and $\diam(S_{1,2})=4$, it follows that
\(
V(D)\subseteq N_D^{\le 4}(C).
\)

A vertex subset $X\subseteq D$ is \emph{scattered} if for any $x,y\in X$, $\dist(x,y)\geq 3$ within $D$. 
We apply the following reduction rule:
\begin{rrule}\label{rule:scattered_S12}
If there exists a vertex $v\in C$ and a scattered set $X\subseteq N_D(v)$ of size at least $5k$, then delete $v$ and decrease $k$ by $1$.
\end{rrule}

\begin{lemma}\label{lem:s12-scattered-safe}
\Cref{rule:scattered_S12} is safe.
\end{lemma}

\begin{proof}
Clearly, if $(G,k)$ is a yes-instance, $(G-v, k-1)$ is a yes-instance.
We focus on the opposite direction.
Suppose that $G-v$ contains $k-1$ vertex-disjoint copies of $S_{1,2}$; let $\mathcal Q$ be such a packing.
Every copy in $\mathcal Q$ intersects $C$.
Since $S_{1,2}$ has six vertices, each copy uses at most five vertices of $D$.
Since $X$ is scattered, the closed neighborhoods $N_D[x]$, $x\in X$, are pairwise disjoint.
Hence each copy in $\mathcal Q$ intersects at most five of these closed neighborhoods, and the whole packing intersects at most $5(k-1)$ of them.
Since $|X|\ge 5k$, there exist distinct vertices $x_1,x_2,x_3\in X$ such that $N_D[x_i]\cap V(\mathcal Q)=\emptyset$ for each $i\in\{1,2,3\}$.

Each $x_i$ lies in a non-trivial component of $D$, so it has a neighbor in $D$; choose neighbors $y_1$ of $x_1$ and $y_2$ of $x_2$.
Then $\{v,x_1,y_1,x_2,y_2,x_3\}$ contains a copy of $S_{1,2}$: use $v$ as the center, use $x_1-y_1$ and $x_2-y_2$ as the two branches of length $2$, and use $x_3$ as the branch of length $1$.
This copy is disjoint from every member of $\mathcal Q$, because $\mathcal Q$ avoids $N_D[x_1]\cup N_D[x_2]\cup N_D[x_3]$ and $v\notin V(G-v)$.
\end{proof}

While computing a maximum scattered set in $N_D(c)$ is hard in general,
for each $c\in C$ we greedily construct a maximal scattered set $X_c\subseteq N_D(c)$, and apply \Cref{rule:scattered_S12} only when $|X_c|\ge 5k$.
Once this process stops, every $c\in C$ satisfies $|X_c|<5k$.

\begin{lemma}\label{lem:direct-s12-separated}
Given an instance $(G,k)$ of \textsc{$S_{1,2}$-Packing}, one can compute in polynomial time an equivalent instance together with an isolating set $C$ and a partition $V(G)=C\uplus Z\uplus I$ witnessing that $C$ is an isolating set, such that $I$ is independent, $E(Z,I)=\emptyset$, $|C|=O(k)$, $|Z|=O(k^2)$, and $|E(C\uplus Z)|=O(k^2)$.
\end{lemma}

\begin{proof}
Fix $c\in C$.
The maximality of $X_c$ implies that every vertex of $N_D(c)$ has $D$-distance at most $2$ from some vertex of $X_c$.
Since $\Delta(D)\le 2$, every radius-$2$ ball in $D$ contains at most five vertices.
Hence $|N_D(c)|\le 5|X_c|<25k$.
As $|C|=O(k)$, it follows that $|N_D(C)|=O(k^2)$.

Set $Z:=V(D)$.
Every vertex of $Z$ lies in $N_D^{\le 4}(C)=N_D^{\le 3}(N_D(C))$.
Since $\Delta(D)\le 2$, every radius-$3$ ball in $D$ contains at most seven vertices, and therefore $|Z|=O(k^2)$.
Also, $|E(Z)|=O(k^2)$ because $\Delta(D)\le 2$, and
\(
|E(Z,C)|=\sum_{c\in C}|N_D(c)|=O(k^2).
\)
Since $|C|=O(k)$, we have $|E(C)|=O(k^2)$ as well, and thus $|E(C\uplus Z)|=O(k^2)$.

Let $I:=V(G)\setminus (C\cup Z)$.
Every vertex of $I$ lies in a trivial component of $G-C$, hence $I$ is independent.
Moreover, no vertex of $I$ is adjacent to a vertex of $Z$, because distinct components of $G-C$ are anticomplete.
Thus $V(G)=C\uplus Z\uplus I$ is the desired partition witnessing that $C$ is an isolating set.
\end{proof}

\shortonly{
\markproofomitted
\begin{lemma}\label{lem:direct-sd11-separated}
Given an instance $(G,k)$ of \textsc{$S_{d_1,1}$-Packing}, one can compute in polynomial time an equivalent instance together with an isolating set $C$ and a partition $V(G)=C\uplus Z\uplus I$ witnessing that $C$ is an isolating set, such that $|C|=O(k)$, $|Z|=O(k^2)$, and $|E(C\uplus Z)|=O(k^2)$, where the hidden constants depend polynomially on $d_1$.
\end{lemma}
This is the first-stage reduction used for the $S_{d_1,1}$ part of \Cref{thm:short-subdivided-stars}; see the full version for the proof.
}

\longonly{
\subsection{Reduction for $S_{d_1,1}$}\label{subsec:direct-sd11}

Throughout this subsection, let $H:=S_{d_1,1}$ for a fixed integer $d_1\ge 1$, and set $d:=d_1+1$.

We show that every instance of \textsc{$H$-Packing} can be reduced to an equivalent instance together with an isolating set $C$ and a partition $V(G)=C\uplus Z\uplus I$ witnessing that $C$ is an isolating set, where $|C|=O(d_1^2k)$ and $|Z|=O(d_1^6k^2)$.

\begin{lemma}\label{lem:sd11-structure}
Let $R$ be a connected $S_{d_1,1}$-free graph with at least two vertices.
Then at least one of the following holds: $R$ is a star or $\Delta(R)\le d$.
\end{lemma}

\begin{proof}
Assume that $R$ is not a star, and $x\in V(R)$ has degree at least $d+1$.
Because $R$ is not a star, there exists an edge $yz$ with $y\in N_R(x)$ and $z\neq x$. 
%
Using $x$ as the center, the edge $yz$ as the branch of length $2$, and any $d-1$ further neighbors of $x$ as the branches of length $1$ gives a copy of $S_{d_1,1}$, a contradiction.
\end{proof}

Let $D$ be the union of all non-star components $R\in\mathcal R$.
Again, a vertex subset $X\subseteq D$ is \emph{scattered} if for any $x,y\in X$, $\dist(x,y)\geq 3$ within $D$. 
We apply the following reduction rule:
\begin{rrule}\label{rule:scattered_Sd1}
If there exist a vertex $v\in C$ and a scattered set $X\subseteq N_D(v)$ of size at least $(d+1)k$, then delete $v$ and decrease $k$ by $1$.
\end{rrule}

\begin{lemma}\label{lem:sd11-scattered-safe}
\Cref{rule:scattered_Sd1} is safe.
\end{lemma}

\begin{proof}
Clearly, if $(G,k)$ is a yes-instance, $(G-v,k-1)$ is a yes-instance.
We focus on the opposite direction.
Suppose that $G-v$ contains $k-1$ pairwise vertex-disjoint copies of $H$; let $\mathcal Q$ be such a packing.
Every copy in $\mathcal Q$ intersects $C$.
Since $H$ has $d+2$ vertices, each copy therefore uses at most $d+1$ vertices of $D$.
Because $X$ is scattered, the closed neighborhoods $N_D[x]$, $x\in X$, are pairwise disjoint, and each copy in $\mathcal Q$ can intersect at most $d+1$ of them.
Consequently the whole packing intersects at most $(d+1)(k-1)$ of these neighborhoods.
Since $|X|\ge (d+1)k$, there exist distinct vertices $x_1,x_2,\dots,x_d\in X$ such that $N_D[x_i]\cap V(\mathcal Q)=\emptyset$ for every $i\in\{1,\dots,d\}$.
Each $x_i$ lies in a non-trivial component of $D$, hence has a neighbor in $D$; choose one neighbor of $x_1$ and denote it by $y_1$.
Then $\{v,x_1,y_1,x_2,\dots,x_d\}$ contains a copy of $S_{d_1,1}$: use $v$ as the center, use $x_1-y_1$ as the unique branch of length $2$, and use $x_2,\dots,x_d$ as the $d-1$ branches of length $1$.
This copy is disjoint from every member of $\mathcal Q$, because $\mathcal Q$ avoids $N_D[x_1]\cup \cdots \cup N_D[x_d]$ and $v\notin V(G-v)$.
\end{proof}

Again, we do not compute a maximum scattered set in $N_D(c)$. For each $c\in C$ we greedily construct a maximal scattered set $X_c\subseteq N_D(c)$, and apply \Cref{rule:scattered_Sd1} only when $|X_c|\ge (d+1)k$.
Once this process stops, every $c\in C$ satisfies $|X_c|<(d+1)k$.
Moreover, by \Cref{rule:irrelevant-element}, every vertex of $D$ belongs to some copy of $H$.
Since each such copy intersects $C$ and $\diam(H)=3$, it follows that $V(D)\subseteq N_D^{\le 3}(C)$.

\begin{lemma}\label{lem:direct-sd11-separated}
Given an instance $(G,k)$ of \textsc{$H$-Packing}, one can compute in polynomial time an equivalent instance together with an isolating set $C$ and a partition $V(G)=C\uplus Z\uplus I$ witnessing that $C$ is an isolating set, such that $I$ is an independent set, $E(Z,I)=\emptyset$, $|C|=O\!\left(d_1^2k\right)$, $|Z|=O\!\left(d_1^6k^2\right)$, and $|E(C\uplus Z)|=O(d_1^7k^2)$.
\end{lemma}

\begin{proof}
If $(G,k)$ is a yes-instance, return the disjoint union of $k$ copies of $H$.
Otherwise, after applying \Cref{rule:component-reduction} with $H=S_{d_1,1}$, let $C_0:=C$.
Then $|C_0|<(d+2)k$ and $|\mathcal R|<d(d+2)k$.

For each star component $R\in \mathcal R$, move the center of $R$ to $C$.
Since there are fewer than $d(d+2)k$ components in $\mathcal R$, we obtain $|C|<|C_0|+d(d+2)k<(d+1)(d+2)k$.
Every remaining non-trivial component of $G-C$ lies in $D$, and by \Cref{lem:sd11-structure}, $\Delta(D)\le d$.

Fix a vertex $c\in C_0$.
The maximality of $X_c$ implies that every vertex of $N_D(c)$ has $D$-distance at most $2$ from some vertex of $X_c$.
Since $\Delta(D)\le d$, every radius-$2$ ball in $D$ contains at most $1+d+d(d-1)=d^2+1$ vertices.
Hence
\(
|N_D(c)|\le (d^2+1)|X_c|<(d^2+1)(d+1)k=(d^3+d^2+d+1)k.
\)

Since vertices outside $N_D^{\le 3}(C)$ were deleted and the vertices moved from $G-C_0$ into $C\setminus C_0$ are centers of star components, every vertex of $D$ lies in $N_D^{\le 3}(C_0)=N_D^{\le 2}(N_D(C_0))$.
Therefore
\(
|V(D)|\le (d^2+1)|N_D(C_0)|\le (d^2+1)|C_0|(d^3+d^2+d+1)k = O(d^6k^2).
\)

Let $Z:=V(D)$ and let $I:=V(G)\setminus (C\cup Z)$.
Then $I$ is independent: vertices outside $C\cup Z$ lie either in trivial components of $G-C_0$ or in star components whose centers were moved to $C$, and in both cases they induce an independent set.
Moreover, no vertex of $I$ is adjacent to a vertex of $Z$, because different components of $G-C_0$ are anticomplete.

Since $G[D]$ has maximum degree at most $d$, we have $\Delta(G[Z])\le d$.
Hence $|E(Z)|=O(d|Z|)=O(d^7k^2)$.
Also,
\(
|E(D,C)|=|E(D,C_0)|\le |C_0|(d^3+1)k = O(d^4k^2),
\)
because vertices of $D$ have no neighbors in $C\setminus C_0$, and
\(
|E(Z,C)|=|E(Z,C_0)|\le |C_0|\cdot |\mathcal R|\cdot (d+1)=O(d^4k^2).
\)
Thus $|E(Z,C)|=O(d^4k^2)$.
Finally, $|E(C)|=O(|C|^2)=O(d^4k^2)$, and therefore $|E(C\uplus Z)|=O(d^7k^2)$.
This yields the desired partition $V(G)=C\uplus Z\uplus I$.
\end{proof}
}

\longonly{
\section{Reducing the General Case to Instances with a Quadratic-Size Vertex Cover}\label{sec:first-stage-vc}
}
\shortonly{
\subsection{Reducing to Instances with a Quadratic-Size Vertex Cover}\label{sec:first-stage-vc}
}

\longonly{
In the previous section, for the special cases $H\in \{S_{0,2}=P_5,S_{1,2},S_{d_1,1}\}$, we obtained an isolating set $C$ of size $O(k)$ together with a set $Z$ of size $O(k^2)$.
We now turn to a generic reduction for subdivided stars.
It transforms every instance of \textsc{$S_{d_1,d_2}$-Packing} into an equivalent instance with a vertex cover of size $O(d^5k^2)$, where $d=d_1+d_2$.
Formally, we prove the following:
}

\shortonly{
Our goal in this section is to reduce the general case of \textsc{$S_{d_1,d_2}$-Packing} to an instance with quadratic-size vertex cover.
Due to space limitations, we give only a technical overview here; see the full version for a detailed account.

As before, we begin by computing a maximal family $\mathcal{P}$ of pairwise vertex-disjoint copies of $S_{d_1,d_2}$ and setting $C:=\bigcup_{P\in \mathcal{P}}V(P)$, and assume $|C|= O(dk)$, where $d=d_1+d_2$.
Let $M$ be a maximal matching in $G-C$, and let $D:=V(M)$ and $I:=V(G)\setminus (C\cup D)$.
Since $M$ is maximal, the set $I$ is independent.
The main strategy is to reduce the size of $D$ to $O_d(k^2)$ using some reduction rules; this way, $C\cup D$ has size $O_d(k^2)$, and thus, is a desired vertex cover.
By applying \Cref{rule:irrelevant-element}, we can assume that all vertices in $D$ have distance at most $\diam(S_{d_1,d_2})=4$ from $C$.
Thus, to bound $|D|$, it suffices to bound the number of $C-D$ paths of length at most $4$, whose inner vertices are in $D\cup I$, by $O_d(k^2)$.
Since $I$ is independent, vertices in $I$ cannot be adjacent on such a path; thus the possible types are
\[
    D,\ DD,\ ID,\ DDD,\ DID,\ IDD,\ DDDD,\ DDID,\ DIDD,\ IDDD,\ IDID.
\]

A seemingly natural strategy for bounding the number of such paths is to bound $\Delta(G)$.
Indeed, from the fact that $G-C$ is $S_{d_1,d_2}$-free, we can bound $|N_D(u)|$ for $u\in D$ by $O(d)$; thus the $D-D$ expansion is bounded. 
Moreover, for $c\in C$, applying a simple reduction rule bounds $|N_D(c)|$ by $O(dk)$; thus the $C-D$ expansion is also bounded.
However, an analogous analysis does not bound $C-I$ and $D-I$ expansions; thus the $\Delta(G)$ itself cannot be bounded. 
Nevertheless, using arguments based on K\H{o}nig's Theorem, we can instead bound two-hop expansions, $C-I-D$ and $D-I-D$, by $O(d^2k)$ and $O(d^2)$, respectively; this still suffices to bound the number of paths ending at $D$. 
Combining these bounds yields the desired bound $|D|=|N_{G-C}^{\leq 4}(C)\cap D|= O(d^5k^2)$.
This leads to the following.
}

\markproofomitted
\begin{lemma}\label{lem:first-stage-vc}
For every fixed subdivided star $H=S_{d_1,d_2}$, where $d:=d_1+d_2$, there exists a polynomial-time reduction that transforms every instance $(G,k)$ of \textsc{$H$-Packing} into an equivalent instance $(G',k')$ such that $k'\le k$ and $\vc(G') = O(d^5k^2)$.
\end{lemma}

\longonly{
\subsection{Degree analysis}

Fix $H:=S_{d_1,d_2}$ and write $d:=d_1+d_2$.
Let $\mathcal P$ be a maximal family of pairwise vertex-disjoint copies of $H$ in $G$, and set $C:=\bigcup_{P\in\mathcal P}V(P)$.
If $|\mathcal P|\ge k$, then $(G,k)$ is a yes-instance, so we may assume that $|\mathcal P|<k$.
Then $|C|<|V(H)|k\le (2d+1)k$, and maximality implies that $G-C$ is $H$-free.

Let $M$ be a maximal matching in $G-C$, and let $D:=V(M)$ and $I:=V(G)\setminus (C\cup D)$.
For each vertex $u\in D$, let $\mu(u)$ denote its mate on the matching edge of $M$.
Since $M$ is maximal, the set $I$ is independent.

We begin by showing that neither $G[D]$ nor the neighborhoods of vertices of $I$ inside $D$ can be too large, since otherwise one would already find a copy of $H$ inside $G-C$.

\begin{lemma}\label{lem:degD}
The graph $G[D]$ has maximum degree at most $2d-1$.
\end{lemma}

\begin{proof}
Suppose for contradiction that some vertex $u\in D$ has at least $2d$ neighbors in $G[D]$, and let $u\mu(u)\in M$ be the matching edge containing $u$.
Then $N_{G[D]}(u)\setminus\{\mu(u)\}$ has at least $2d-1$ vertices.
Since each edge of $M\setminus\{u\mu(u)\}$ contributes at most two vertices to this set, there exist $d$ distinct edges $e_1,\dots,e_d\in M\setminus\{u\mu(u)\}$ such that, for each $j$, the edge $e_j$ has an endpoint $x_j$ adjacent to $u$.
Let $y_j$ be the other endpoint of $e_j$.
Because $M$ is a matching, the vertices $x_1,y_1,\dots,x_d,y_d$ are pairwise distinct and all different from $u$.

Now use $u$ as the center: for $j\le d_2$, the path $u-x_j-y_j$ forms a branch of length $2$, and for $d_2<j\le d$, the edge $u-x_j$ forms a branch of length $1$.
This yields a copy of $H$ in $G-C$, contrary to the choice of $C$.
Hence $\Delta(G[D])\le 2d-1$.
\end{proof}

The same idea also bounds how many vertices of $D$ a vertex of $I$ can see.
\begin{lemma}\label{lem:ItoD}
For every vertex $x\in I$, $|N_G(x)\cap D|\le 2d-2$.
\end{lemma}

\begin{proof}
Suppose for contradiction that some vertex $x\in I$ satisfies $|N_G(x)\cap D|\ge 2d-1$.
Since each edge of $M$ contributes at most two vertices to $N_G(x)\cap D$, there exist $d$ distinct edges $e_1,\dots,e_d\in M$ such that, for each $j$, the edge $e_j$ has an endpoint $v_j$ adjacent to $x$.
Let $w_j$ be the other endpoint of $e_j$.
Because $M$ is a matching, the vertices $v_1,w_1,\dots,v_d,w_d$ are pairwise distinct and all different from $x$.

Now use $x$ as the center: for $j\le d_2$, the path $x-v_j-w_j$ forms a branch of length $2$, and for $d_2<j\le d$, the edge $x-v_j$ forms a branch of length $1$.
This yields a copy of $H$ in $G-C$, contrary to the choice of $C$.
\end{proof}

We now introduce a reduction rule that controls how many matching edges can survive near a fixed vertex of $C$.
If too many such edges survive around some $c\in C$, then $c$ can be forced into one copy of $H$ and removed.
For $c\in C$, define $F_c^{\circ}:=\{e\in E(G-C): e\cap N_G(c)\neq \emptyset\}$.

\begin{rrule}\label{rule:matching}
If $G[F_c^{\circ}]$ contains a matching of size $2dk$, delete $c$ and decrease $k$ by $1$.
\end{rrule}

\begin{lemma}\label{lem:matching-rule-safe}
\Cref{rule:matching} is safe.
\end{lemma}

\begin{proof}
If $(G,k)$ is a yes-instance, then clearly $(G-c,k-1)$ is a yes-instance.
We prove the converse.

Let $N$ be a matching of size $2dk$ in $G[F_c^{\circ}]$, and assume that $(G-c,k-1)$ is a yes-instance.
Let $\mathcal Q$ be an $H$-packing of size $k-1$ in $G-c$.
Every copy in $\mathcal Q$ intersects $C$, and since $c\notin G-c$, in fact every copy intersects $C\setminus\{c\}$.
Hence each copy contains at most $|V(H)|-1\le 2d$ vertices of $G-C$, so the whole packing uses at most $2d(k-1)$ vertices of $G-C$.

Since $N$ is a matching in $G-C$, each vertex of $G-C$ lies on at most one edge of $N$.
Therefore the vertices of $V(\mathcal Q)\setminus C$ block at most $2d(k-1)$ edges of $N$.
As $|N|=2dk$, at least $2d$ edges of $N$ are disjoint from $V(\mathcal Q)$.
Choose $d$ such edges and denote them by $f_1,\dots,f_d$.

Since each edge $f_j$ belongs to $F_c^{\circ}$, it has an endpoint $x_j\in N_G(c)$; let $y_j$ be the other endpoint of $f_j$.
Because the edges $f_1,\dots,f_d$ are pairwise distinct edges of the matching $N$, the vertices $x_1,y_1,\dots,x_d,y_d$ are pairwise distinct and lie outside $V(\mathcal Q)\cup C$.

Now use $c$ as the center: for $j\le d_2$, the path $c-x_j-y_j$ forms a branch of length $2$, and for $d_2<j\le d$, the edge $c-x_j$ forms a branch of length $1$.
This yields a copy of $H$ disjoint from $\mathcal Q$, and therefore $(G,k)$ is a yes-instance.
\end{proof}

Henceforth, we may assume that \Cref{rule:matching} no longer applies to any vertex of $C$.

\subsection{Size analysis}

Our remaining task is to bound the part of $D$ that can still matter after \Cref{rule:matching} no longer applies.
By exhaustively applying \Cref{rule:irrelevant-element}, we may assume that every vertex of $D$ belongs to some copy of $H$ in $G$.
Since every copy of $H$ intersects $C$ and $\diam(H)\le 4$, it follows that every vertex of $D$ lies at distance at most $4$ from some vertex of $C$.
We therefore bound the set of vertices of $D$ that can be reached from $C$ by a path of length at most $4$.
Steps that stay inside $D$ contribute only a constant factor $O(d)$, by \Cref{lem:degD}.
The initial transitions $c-D$ and $c-I-D$ are handled separately in \Cref{lem:Ac,lem:Sc}, and after one has reached $D$, the only potentially non-constant further expansion is a step of the form $D-I-D$, which is controlled by \Cref{lem:Zu}.

We start with bounding the neighbors of $c$ that already lie in $D$.
\begin{lemma}\label{lem:Ac}
For $c\in C$, let $A_c:=N_G(c)\cap D$.
Then $|A_c|<4dk$.
\end{lemma}

\begin{proof}
Let $M_c:=\{e\in M : e\cap A_c\neq\emptyset\}$.
Then every edge of $M_c$ lies in $E(G-C)$ and has an endpoint in $N_G(c)$, so $M_c\subseteq F_c^{\circ}$.
Since $M_c$ is a matching and \Cref{rule:matching} no longer applies to~$c$, we have $|M_c|<2dk$.
Therefore $|A_c|\le 2|M_c|<4dk$.
\end{proof}

For each vertex $a\in C\cup D$, let
\(
T_a:=N_G(N_G(a)\cap I)\cap D.
\)
Thus $T_a$ is the set of vertices of $D$ that can be reached from $a$ by first taking an edge to a vertex of $I$ and then an edge back to $D$.

The next lemma handles the vertices of $D$ reached from $c$ through one intermediate vertex of $I$.
\begin{lemma}\label{lem:Sc}
For $c\in C$, let $S_c:=T_c$.
Then $|S_c|<(4d^2-2d)k$.
\end{lemma}

\begin{proof}
Let $L_c:=N_G(c)\cap I$.
Then $S_c=N_G(L_c)\cap D$.
Consider the bipartite graph $B_c:=G[L_c,S_c]$.
If $B_c$ contained a matching of size $2dk$, then those $2dk$ edges would form a matching of size $2dk$ in $G[F_c^{\circ}]$, because every such edge lies in $E(G-C)$ and has its endpoint in $L_c\subseteq N_G(c)$.
This would contradict the fact that \Cref{rule:matching} no longer applies to $c$.
Hence $B_c$ has no matching of size $2dk$.

Since $B_c$ is bipartite, K\H{o}nig's theorem yields a vertex cover $X_c$ of $B_c$ with $|X_c|<2dk$.
By \Cref{lem:ItoD}, every vertex $x\in L_c\subseteq I$ satisfies $\deg_{B_c}(x)\le |N_G(x)\cap D|\le 2d-2$.
Because $X_c$ meets every edge of $B_c$, we have $|S_c|\le |X_c\cap S_c|+\sum_{x\in X_c\cap L_c}\deg_{B_c}(x)<2dk+(2d-2)\cdot 2dk=(4d^2-2d)k$.
\end{proof}

Finally, we bound one further alternation of the form $D-I-D$ starting from a fixed vertex of $D$.
\begin{lemma}\label{lem:Zu}
Fix $u\in D$, and let $Z_u:=T_u$.
Then $|Z_u|\le 2d^2$.
\end{lemma}

\begin{proof}
Let $L_u:=N_G(u)\cap I$ and let $R_u:=Z_u\setminus\{u\}$.
Consider the bipartite graph $B_u:=G[L_u,R_u]$.
If $B_u$ contained a matching of size $d$, say $x_1v_1,\dots,x_dv_d$, then the vertices $x_1,\dots,x_d$ would be pairwise distinct, the vertices $v_1,\dots,v_d$ would be pairwise distinct and all different from $u$, and all these vertices would lie in $G-C$.
Using $u$ as the center, for $i\le d_2$ the path $u-x_i-v_i$ forms a branch of length $2$, and for $d_2<i\le d$ the edge $u-x_i$ forms a branch of length $1$.
This yields a copy of $H$ in $G-C$, contrary to the choice of $C$.
Hence $B_u$ has no matching of size $d$.

Again by K\H{o}nig's theorem, $B_u$ has a vertex cover $Y_u$ of size at most $d-1$.
By \Cref{lem:ItoD}, every vertex $x\in L_u$ satisfies $|N_G(x)\cap D|\le 2d-2$, and one of these neighbors is $u$.
Therefore $\deg_{B_u}(x)\le 2d-3$ for every $x\in L_u$.
Since $Y_u$ meets every edge of $B_u$, we obtain
\(
|R_u|\le |Y_u\cap R_u|+\sum_{x\in Y_u\cap L_u}\deg_{B_u}(x)\le (d-1)+(2d-3)(d-1)=2(d-1)^2.
\)
Since $Z_u\subseteq R_u\cup\{u\}$, we obtain $|Z_u|\le 1+|R_u|\le 1+2(d-1)^2\le 2d^2$.
\end{proof}

We can now combine the previous three bounds and control all vertices of $D$ that may lie at distance at most $4$ from a fixed vertex of $C$.
\begin{lemma}\label{lem:Uc}
For $c\in C$, let $U_c$ be the set of vertices $v\in D$ for which there exists a path from $c$ to $v$ of length at most $4$ whose internal vertices lie in $D\cup I$.
Then $|U_c|<88d^4k$.
\end{lemma}

\begin{proof}
By \Cref{lem:degD}, for every $X\subseteq D$ and every $r\ge 0$, $|N_D^{\le r}(X)|\le \bigl(1+(2d-1)+\cdots +(2d-1)^r\bigr)|X|$.

Since $I$ is independent, every such path alternates between vertices of $D$ and vertices of $I$ after it leaves $c$.
Thus the possible types are
\[
D,\ DD,\ ID,\ DDD,\ DID,\ IDD,\ DDDD,\ DDID,\ DIDD,\ IDDD,\ IDID.
\]
We group these types into four families.

If the path contains no vertex of $I$, then its type is one of $D,DD,DDD,DDDD$.
Its first vertex after $c$ lies in $A_c=N_G(c)\cap D$, and the rest of the path stays inside $D$.
Hence every endpoint of this form lies in $N_D^{\le 3}(A_c)$.

If the path starts with a vertex of $I$ and contains no further vertex of $I$, then its type is one of $ID,IDD,IDDD$.
The first vertex of $D$ reached after this initial step lies in $S_c$, and the remainder of the path stays inside $D$.
Hence every endpoint of this form lies in $N_D^{\le 2}(S_c)$.

Next consider the types $DID,DDID,IDID$.
Let $u$ be the vertex of $D$ immediately before the last occurrence of $I$ on the path.
Then $u\in A_c$ for type $DID$, $u\in N_D(A_c)$ for type $DDID$, and $u\in S_c$ for type $IDID$.
In each case the final $I-D$ step ends at a vertex of $Z_u$.
Hence every endpoint of these three types lies in $\bigcup_{u\in A_c\cup N_D(A_c)\cup S_c} Z_u$.

The only remaining type is $DIDD$.
Such a path has the form $c-u-x-v-w$ with $u\in A_c$, $x\in I$, and $v,w\in D$.
Since $v\in Z_u$, the endpoint $w$ lies in $N_D(Z_u)$.
Hence every endpoint of this type lies in $\bigcup_{u\in A_c} N_D(Z_u)$.

Consequently,
\[
U_c\subseteq N_D^{\le 3}(A_c)\cup N_D^{\le 2}(S_c)\cup \bigcup_{u\in A_c\cup N_D(A_c)\cup S_c} Z_u \cup \bigcup_{u\in A_c} N_D(Z_u).
\]

We now bound the size of each set on the right-hand side.
Since $1+(2d-1)+(2d-1)^2+(2d-1)^3\le 8d^3$ and $1+(2d-1)+(2d-1)^2\le 4d^2$, \Cref{lem:Ac,lem:Sc} give
\[
|N_D^{\le 3}(A_c)|\le 8d^3|A_c|<8d^3\cdot 4dk=32d^4k
\]
and
\[
|N_D^{\le 2}(S_c)|\le 4d^2|S_c|<4d^2\cdot (4d^2-2d)k<16d^4k.
\]

Also, \Cref{lem:Ac} implies that
\[
|N_D(A_c)|\le (2d-1)|A_c|<(2d-1)\cdot 4dk=(8d^2-4d)k.
\]
Hence, using \Cref{lem:Ac,lem:Sc,lem:Zu},
\[
|A_c|+|N_D(A_c)|+|S_c|<4dk+(8d^2-4d)k+(4d^2-2d)k=(12d^2-2d)k,
\]
so
\[
\left|\bigcup_{u\in A_c\cup N_D(A_c)\cup S_c} Z_u\right|
\le \sum_{u\in A_c\cup N_D(A_c)\cup S_c}|Z_u|
<(12d^2-2d)k\cdot 2d^2
<24d^4k.
\]
Finally, for every $u\in A_c$,
\[
|N_D(Z_u)|\le (2d-1)|Z_u|\le (2d-1)\cdot 2d^2<4d^3,
\]
whence
\[
\left|\bigcup_{u\in A_c} N_D(Z_u)\right|
\le \sum_{u\in A_c}|N_D(Z_u)|
<|A_c|\cdot 4d^3
<4dk\cdot 4d^3
=16d^4k.
\]
Combining the four bounds yields
\[
|U_c|<32d^4k+16d^4k+24d^4k+16d^4k=88d^4k.\qedhere
\]
\end{proof}
Since every vertex of $D$ lies in $U_c$ for some $c\in C$, it follows that $|D|=O(d^5k^2)$.
Together with \Cref{rule:irrelevant-element} and the bound $|C| = O(dk)$, \Cref{lem:Uc} yields \Cref{lem:first-stage-vc}.
}

\section{Kernelization from an Isolating Set}\label{sec:bounded-vc-separated}

The goal of this section is to reduce instances of \textsc{$H$-Packing} given together with an isolating set to instances of a bounded number of vertices, thereby completing the kernelization.

\subsection{General reduction}

Here, we give an algorithm that kernelizes the given instance $(G,k)$ with an isolating set $C$.
We work with the general $H$; we begin by defining the following \emph{$r$-replaceability}, which is a condition that depends solely on the graph $H$.

\begin{definition}
Let $H$ be a graph, let $r\ge 1$ be an integer, and let $x=(x_1,\dots,x_r)$ be a tuple of positive integers.
For a vertex cover $C_H$ of $H$, put $I_H:=V(H)\setminus C_H$.
The \emph{$x$-extension} of $H$ with respect to $C_H$ is a graph obtained from $H[C_H]$ as follows:
For every $i\in \{1,\dots,r\}$, let $\mathcal{T}_i\subseteq \tbinom{C_H}{i}$ be the family of size-$i$ subsets $T$ of $C_H$ such that there exists a vertex $u\in I_H$ with $T\subseteq N_H(u)$.
For each $T\in \mathcal{T}_i$, add $x_i$ fresh vertices adjacent exactly to the vertices of~$T$.
We say that $H$ is \emph{$r$-replaceable}, witnessed by $x$, if for every vertex cover $C_H$ of $H$, the $x$-extension contains a copy of $H$.
\end{definition}

The goal here is to prove the following.
\begin{theorem}\label{thm:replaceable_kernel}
Assume $H$ is $r$-replaceable witnessed by $x=(x_1,\dots, x_r)$.
There is a polynomial-time algorithm that, given an instance $(G,k)$ with a partition $V(G)=C\uplus I\uplus Z$, which witnesses that $C$ is an isolating set, outputs an equivalent instance $(G',k)$ such that
\[
    |V(G')|\leq |Z|+|C|+\sum_{i=1}^{r}x_i\binom{|C|}{i} = |Z|+O(|C|^r),
    \qquad |E(G')|\leq |E(G[Z\cup C])|+O(|C|^{r+1}).
\]
\end{theorem}

Let $(G,k)$ be an instance of \textsc{$H$-Packing} with a partition $V(G)=C\uplus I\uplus Z$, which witnesses that $C$ is an isolating set.
Let
\[
L:=\biguplus_{i=1}^r \bigl(\tbinom{C}{i}\times [x_i]\bigr),
\]
where $[x]=\{1,\dots,x\}$.
For a vertex $\lambda=(D,j)\in \tbinom{C}{i}\times [x_i]$, write $D(\lambda):=D$.
Define the bipartite graph $\mathcal B=(L\uplus I,F)$ by letting $\lambda v\in F$ if and only if $D(\lambda)\subseteq N_G(v)$.

Apply \Cref{lem:expansion} to $\mathcal B$ with $q=1$.
We obtain subsets $\widehat L\subseteq L$, $\widehat I\subseteq I$, and an edge set $M\subseteq F$ such that:
\begin{enumerate}
    \item $|\widehat I|\le |\widehat L|$,
    \item every vertex of $L\setminus \widehat L$ is incident with exactly one edge of $M$,
    \item every vertex of $I\setminus \widehat I$ is incident with at most one edge of $M$,
    \item there is no edge of $\mathcal B$ between $\widehat L$ and $I\setminus \widehat I$.
\end{enumerate}
For every $\lambda\in L\setminus \widehat L$, let $m_\lambda\in I\setminus \widehat I$ be the unique vertex satisfying $\lambda m_\lambda\in M$.
For every $D\in \bigcup_{i=1}^r \tbinom{C}{i}$, define
\[
R_D:=\{m_\lambda : \lambda\in L\setminus \widehat L\text{ and }D(\lambda)=D\}.
\]
Thus $R_D$ is the set of such vertices $m_\lambda$ whose left endpoint $\lambda$ satisfies $D(\lambda)=D$.
Let $R:=\bigcup_{D\in \bigcup_{i=1}^r \tbinom{C}{i}} R_D$.
We apply the following rule.

\begin{rrule}\label{rule:replaceable-isolating}
Replace $(G,k)$ by $(G',k)$, where $G':=G[C\cup Z\cup \widehat I\cup R]$.
\end{rrule}

\begin{lemma}\label{lem:replaceable-isolating-size}
The graph $G'$ from \Cref{rule:replaceable-isolating} satisfies
\[
|V(G')|\le |Z|+|C|+\sum_{i=1}^r x_i\binom{|C|}{i}, 
\qquad |E(G')|\leq |E(G[Z\cup C])|+|C|\sum_{i=1}^r x_i\binom{|C|}{i}.
\]
\end{lemma}

\begin{proof}
By the first property, $|\widehat I|\le |\widehat L|$.
Also, $|R|=|L\setminus \widehat L|$.
Therefore $|\widehat I|+|R|\le |\widehat L|+|L\setminus \widehat L|=|L|$.
Since $|L|=\sum_{i=1}^r x_i\binom{|C|}{i}$, we obtain
\[
|V(G')|=|Z|+|C|+|\widehat I|+|R|\le |Z|+|C|+\sum_{i=1}^r x_i\binom{|C|}{i}.
\]
The bound of $|E(G')|$ is straightforward.
\end{proof}

Now we prove that the reduction is safe.
Since $G'\subseteq G$, it is clear that if $(G',k)$ is a yes-instance, then $(G,k)$ is also a yes-instance.
Thus, it remains to prove the converse.
The idea is to show that each copy $P$ in a packing of $G$ can be replaced by a copy contained in a set $\Ext(P)\subseteq V(G')$ (\Cref{lem:Ext-contains-copy}), defined later, and that these sets are pairwise disjoint for different copies (\Cref{lem:Ext-disjoint}).
Fix a packing $\mathcal P$ of $k$ copies of $H$ in $G$.
For each $P\in \mathcal{P}$, define
\[
    \mathsf{Ext}(P):=(V(P)\cap (C\cup Z\cup \widehat{I}))\cup \bigcup_{D\in \bigcup_{i=1}^{r}\binom{V(P)\cap C}{i}}R_D.
\]

\begin{lemma}\label{lem:Ext-disjoint}
The sets $\Ext(P)$, for $P\in \mathcal P$, are pairwise disjoint.
\end{lemma}

\begin{proof}
Let $P,Q\in \mathcal P$ be distinct.
Since the copies in $\mathcal P$ are vertex-disjoint, we immediately have that $V(P)\cap (C\cup Z\cup \widehat I)$ and $V(Q)\cap (C\cup Z\cup \widehat I)$ are
disjoint.
Thus, it remains to show that no vertex of $R$ belongs to both $\Ext(P)$ and $\Ext(Q)$.

First, observe that the sets $R_D$ are pairwise disjoint.
Indeed, if some vertex $r\in R_D\cap R_{D'}$, then $r=m_\lambda=m_{\lambda'}$ for two vertices $\lambda,\lambda'\in L\setminus \widehat L$ with $D(\lambda)=D$ and $D(\lambda')=D'$.
But then the vertex $r\in I\setminus \widehat I$ is incident with two edges of $M$, namely $\lambda r$ and $\lambda' r$, contradicting \textup{(E3)}.
Hence $D=D'$ and $R_D\cap R_{D'}=\emptyset$ for $D\neq D'$.

Now suppose for contradiction that some vertex $r$ belongs to both $\Ext(P)$ and $\Ext(Q)$.
Then $r\in R_D$ for some nonempty set $D\subseteq V(P)\cap C$, and also $r\in R_{D'}$ for some nonempty set $D'\subseteq V(Q)\cap C$.
By pairwise disjointness of the mate sets, we have $D=D'$.
Therefore, $D\subseteq V(P)\cap C\cap V(Q)\cap C$,
which is impossible because $P$ and $Q$ are vertex-disjoint and $D\neq \emptyset$.
Thus $\Ext(P)$ and $\Ext(Q)$ are disjoint.
\end{proof}

\begin{lemma}\label{lem:Ext-contains-copy}
For every $P\in \mathcal P$, the set $\Ext(P)$ contains a copy of $H$.
\end{lemma}

\begin{proof}
Fix $P\in \mathcal P$.
Put $C_P:=V(P)\cap (C\cup Z\cup \widehat I)$ and $I_P:=V(P)\setminus C_P=V(P)\cap (I\setminus \widehat I)$.
Then $C_P$ is a vertex cover of $P$ because every vertex of $I_P$ lies in $I$ and $I$ is independent.

Consider the $x$-extension of the copy $P$ with respect to the vertex cover $C_P$.
For every $i\in\{1,\dots,r\}$, let $\mathcal T_i\subseteq \binom{C_P}{i}$ be the family of size-$i$ subsets $T$ of $C_P$ such that there exists a vertex $u\in I_P$ with $T\subseteq N_P(u)$.
We claim that this $x$-extension is a subgraph of $G[\Ext(P)]$.

First, the base graph $P[C_P]$ is a subgraph of $G[\Ext(P)]$, because $C_P\subseteq \Ext(P)$.
It remains to realize the fresh vertices of the extension.
Fix $i\in\{1,\dots,r\}$ and a set $T\in \mathcal T_i$.
By definition, there exists a vertex $u\in I_P\subseteq I\setminus \widehat I$ with $T\subseteq N_P(u)$.
Since $u\in I$, every neighbor of $u$ in $P$ lies in $C$.
Consequently, $T\subseteq V(P)\cap C$.
For each $j\in [x_i]$, the left vertex $\lambda=(T,j)$ belongs to $L$ and is adjacent to $u$ in $\mathcal B$, because $T\subseteq N_G(u)$.
Since $u\in I\setminus \widehat I$, Property~\textup{(E4)} implies that none of these vertices $(T,j)$ lies in $\widehat L$.
Therefore, all $x_i$ vertices $(T,j)$ have mates, and the corresponding mates are pairwise distinct by Property~\textup{(E3)}.
Hence $|R_T|=x_i$.
Moreover, each vertex $r\in R_T$ is adjacent in $G$ to all vertices of $T$, because $r=m_\lambda$ for some $\lambda=(T,j)$ and $\lambda r\in F$ implies $T\subseteq N_G(r)$.
Since $T\subseteq V(P)\cap C$, we have $R_T\subseteq \Ext(P)$.
Thus, the $x_i$ fresh vertices required for $T$ in the $x$-extension are present inside $\Ext(P)$.

We have proved that the $x$-extension of $P$ with respect to $C_P$ is a subgraph of $G[\Ext(P)]$.
Because $P$ is a copy of $H$ and $H$ is $r$-replaceable, witnessed by $x$, this $x$-extension contains a copy of $H$.
Therefore, $\Ext(P)$ contains a copy of $H$.
\end{proof}

\begin{proof}[Proof of \Cref{thm:replaceable_kernel}]
By Lemma~\ref{lem:Ext-contains-copy}, for every $P\in \mathcal P$ we may choose a copy $P^{\star}$ of $H$ inside $G[\Ext(P)]$.
By Lemma~\ref{lem:Ext-disjoint}, the sets $\Ext(P)$ are pairwise disjoint, so the chosen copies $P^{\star}$ are pairwise vertex-disjoint.
Since every $\Ext(P)$ is contained in $V(G')$, this yields a packing of $k$ copies of $H$ in $G'$.
Thus $(G,k)$ is a yes-instance only if $(G',k)$ is a yes-instance.
Together with the trivial reverse implication, this proves that $(G,k)$ and $(G',k)$ are equivalent.
The construction is polynomial-time, and the size bound follows from Lemma~\ref{lem:replaceable-isolating-size}.
\end{proof}

%
%

\subsection{Replaceability Analysis}

In this subsection, we state that $P_5$ and $S_{d_1,d_2}$ with $d_1\ge 1$ are $2$-replaceable, whereas $S_{0,d}$ is not $\ell$-replaceable for any $\ell<d$.
We begin with $P_5$.

\markproofomitted
\begin{lemma}\label{lem:P5_replaceable}
The graph $P_5$ is $2$-replaceable, witnessed by $(1,1)$.
\end{lemma}
\longonly{
\begin{proof}
Let $H$ be a copy of $P_5$, let $C_H$ be an arbitrary vertex cover of $H$, and let $I_H:=V(H)\setminus C_H$.
Consider the $(1,1)$-extension $H^{\mathrm{ext}}$ of $H$ with respect to $C_H$.
We show that $H^{\mathrm{ext}}$ contains a copy of $H$.

If $I_H=\emptyset$, then $H[C_H]=H$, so there is nothing to prove.
Assume therefore that $I_H\neq\emptyset$.
Since $C_H$ is a vertex cover, the set $I_H$ is independent.
Hence, every vertex $u\in I_H$ has all its neighbors in $C_H$, and therefore $N_H(u)$ is either a one-element subset or a two-element subset of $C_H$.
By the definition of the $(1,1)$-extension, for every $u\in I_H$ there exists a fresh vertex $u'$ in $H^{\mathrm{ext}}$ adjacent exactly to the vertices of $N_H(u)$.

We claim that the sets $N_H(u)$ are pairwise distinct for distinct vertices $u\in I_H$.
Indeed, if two distinct vertices of the path had the same one-element neighborhood, then the path would have two distinct leaves adjacent to the same vertex, which is impossible in $P_5$.
If two distinct vertices had the same two-element neighborhood, then the path would contain a cycle, again impossible.
Thus, the vertices $u'$ chosen for distinct $u\in I_H$ are pairwise distinct.

Now keep every vertex of $C_H$, and for each vertex $u\in I_H$ replace $u$ by the corresponding fresh vertex $u'$.
Since $u'$ is adjacent exactly to $N_H(u)$, every edge of $H$ incident with $u$ is reproduced after the replacement.
Therefore, the graph induced by $C_H$ together with all chosen vertices $u'$ is isomorphic to $H=P_5$.
This proves that $P_5$ is $2$-replaceable, witnessed by $(1,1)$.
\end{proof}
}

We extend this to all subdivided stars with at least one branch of length~$1$.
The presence of such a branch is what makes this case $2$-replaceable.

\begin{lemma}\label{lem:Sd1d2_replaceable}
Let $d_1\ge 1$ and $d_2\ge 0$.
Then the subdivided star $S_{d_1,d_2}$ is $2$-replaceable, witnessed by $(d_1,1)$.
\end{lemma}
\begin{proof}
Let $H=S_{d_1,d_2}$, let $c$ be the center of $H$, let $C_H$ be an arbitrary vertex cover of $H$, and let $I_H:=V(H)\setminus C_H$.
Consider the $(d_1,1)$-extension $H^{\mathrm{ext}}$ of $H$ with respect to $C_H$.
We prove that $H^{\mathrm{ext}}$ contains a copy of $H$.

\proofsubparagraph{Case 1: $c\in C_H$.}
In this case, every vertex of $I_H$ is a non-center.
For each $u\in I_H$, put $D_u:=N_H(u)$.
Since $I_H$ is independent, every vertex of $I_H$ has all neighbors in $C_H$.
As $u$ is a non-center of a subdivided star, the set $D_u$ is either a one-element subset of $C_H$ or a two-element subset of $C_H$.

We now examine which traces $D_u$ can repeat.
A singleton trace can repeat only if it is $\{c\}$, because the only vertices of $H$ with neighborhood $\{c\}$ are the leaves on branches of length $1$.
This happens at most $d_1$ times.
Every other singleton trace is of the form $\{b\}$, where $b$ is the internal vertex of a branch of length $2$, and such a trace occurs exactly once.
Similarly, every two-element trace is of the form $\{c,\ell\}$, where $\ell$ is the leaf of a branch of length $2$, and this trace also occurs exactly once.
Thus, the only repeated trace is $\{c\}$, and it occurs at most $d_1$ times.

By the definition of the $(d_1,1)$-extension, for every singleton trace $D$ there are $d_1$ fresh vertices adjacent exactly to $D$, and for every two-element trace $D$ there is one fresh vertex adjacent exactly to $D$.
Therefore, we can choose, for every vertex $u\in I_H$, a distinct fresh vertex $u'$ adjacent exactly to $D_u=N_H(u)$.
Keep all vertices of $C_H$, and replace every $u\in I_H$ by $u'$.
Since $u'$ is adjacent exactly to $N_H(u)$, every edge of $H$ incident with $u$ is reproduced.
Hence, the resulting graph is isomorphic to $H$.

\proofsubparagraph{Case 2: $c\in I_H$.}
Then all neighbors of $c$ lie in $C_H$ because $C_H$ is a vertex cover.
Let $A:=N_H(c)$.
Since $H=S_{d_1,d_2}$, we have $|A|=d_1+d_2$.
Fix any vertex $a\in A$.
Because $\{a\}\subseteq N_H(c)$, the set $\{a\}$ belongs to $\mathcal T_1$.
Hence the $(d_1,1)$-extension contains $d_1$ fresh vertices adjacent exactly to $a$.
Denote them by $x_1,\dots,x_{d_1}$.
Now let $a'\in A\setminus\{a\}$.
Since $\{a,a'\}\subseteq N_H(c)$, the pair $\{a,a'\}$ belongs to $\mathcal T_2$.
Therefore, the extension contains a fresh vertex $y_{a'}$ adjacent exactly to $a$ and $a'$.
Choose any $d_2$ vertices
$a'_1,\dots,a'_{d_2}\in A\setminus\{a\}$
(which is possible because $|A\setminus\{a\}|=d_1+d_2-1\ge d_2$, as $d_1\ge 1$).
Then the subgraph of $H^{\mathrm{ext}}$ on the vertex set
$\{a\}\cup \{x_1,\dots,x_{d_1}\}\cup \{y_{a'_1},\dots,y_{a'_{d_2}}\}\cup \{a'_1,\dots,a'_{d_2}\}$
is isomorphic to $S_{d_1,d_2}$: the center is $a$, the vertices $x_1,\dots,x_{d_1}$ form the $d_1$ branches of length $1$, and for each $j\in\{1,\dots,d_2\}$ the path $a-y_{a'_j}-a'_j$ is a branch of length $2$.

In both cases, the extension contains a copy of $H$.
Therefore $S_{d_1,d_2}$ is $2$-replaceable, witnessed by $(d_1,1)$.
\end{proof}

We now turn to the line $d_1=0$ and show that this phenomenon disappears.
In that case the replaceability parameter cannot be pushed below~$d$.

\begin{proposition}\label{prop:S0d-not-replaceable}
For $d\ge 3$, $S_{0,d}$ is not $\ell$-replaceable for any integer $\ell<d$.
\end{proposition}

\begin{proof}
Fix an integer $\ell<d$, and let $x=(x_1,\dots,x_\ell)$ be any tuple of positive integers.
We prove that $x$ does not witness $\ell$-replaceability of $H$.

Let $c$ be the center of $H$, let $b_1,\dots,b_d$ be the $d$ neighbors of $c$, and let $a_1,\dots,a_d$ be the leaves, where each $a_i$ is adjacent only to $b_i$.
Consider the vertex cover $C_H:=\{b_1,\dots,b_d\}$.
Then $H[C_H]$ is an independent set on $d$ vertices.
Let $X$ be the $x$-extension of $H$ with respect to $C_H$.
By definition of $x$-extension, 1) every added vertex of $X$ has degree at most $\ell$, and 2) the graph $X$ is bipartite with bipartition $(C_H, V(X)\setminus C_H)$.

Suppose for contradiction that $X$ contains a copy $F$ of $S_{0,d}$, and let $z$ be the center of $F$.
From 1) and $\ell < d$, we have $z\in C_H$.
By 2), every neighbor of $z$ in $F$ lies outside $C_H$.
Since $F$ has no branches of length $1$, each such neighbor must in turn have another neighbor in $F$, distinct from $z$.
Again by 2), this second vertex must belong to $C_H\setminus\{z\}$.
Thus each of the $d$ branches of length $2$ in $F$ ends in a distinct vertex of $C_H\setminus\{z\}$.
But $|C_H\setminus\{z\}|=d-1$,
so there are not enough vertices in $C_H\setminus\{z\}$ to realize all $d$ branches, which is a contradiction.
\end{proof}

\shortonly{
\subsection{Deriving \Cref{thm:short-subdivided-stars,thm:general-result}.}
\Cref{thm:short-subdivided-stars,thm:general-result} are obtained by pairing the first-stage reductions with \Cref{thm:replaceable_kernel}.
For $P_5$, \Cref{lem:direct-p5-vc} gives a vertex cover of size $O(k)$; for $S_{1,2}$ and $S_{d_1,1}$, the first stage gives an isolating set $C$ of size $O(k)$ and a set $Z$ of size $O(k^2)$.
Together with the $2$-replaceability results in \Cref{lem:P5_replaceable,lem:Sd1d2_replaceable}, this yields the claimed bounds.
For \Cref{thm:general-result}, the first stage gives a vertex cover of size $O(d^5k^2)$ for \textsc{$S_{d_1,d_2}$-Packing} with $d_1\ge 1$.
Combining this with \Cref{thm:replaceable_kernel} and \Cref{lem:Sd1d2_replaceable} yields a kernel with $O(d^{10}k^4)$ vertices and $O(d^{15}k^6)$ edges.
}

\longonly{
\subsection{Proofs of Main Theorems}\label{sec:consequences}

Here, we derive the explicit kernel bounds for subdivided stars stated in the introduction.
We begin with the three special cases $P_5$, $S_{1,2}$, and $S_{d_1,1}$; together they imply \Cref{thm:short-subdivided-stars}.
\ShortSubdividedStarKernel*

The theorem follows immediately from the following three propositions.

\begin{proposition}\label{thm:P5-result}
\textsc{$P_5$-Packing} admits a polynomial kernel with $O(k^2)$ vertices and $O(k^3)$ edges.
\end{proposition}
\begin{proof}
\Cref{lem:direct-p5-vc} reduces an arbitrary instance $(G,k)$ to an equivalent instance together with a vertex cover $C$ of size $O(k)$. Since $P_5$ is $2$-replaceable witnessed by $(1,1)$ from \Cref{lem:P5_replaceable}, \Cref{thm:replaceable_kernel} reduces it to an equivalent instance on $2|C|+\binom{|C|}{2}=O(k^2)$ vertices and $|C|^2+|C|\cdot \left(2|C|+\binom{|C|}{2}\right)=O(k^3)$ edges.
\end{proof}

\begin{proposition}\label{thm:S12-result}
\textsc{$S_{1,2}$-Packing} admits a polynomial kernel with $O(k^2)$ vertices and $O(k^3)$ edges.
\end{proposition}
\begin{proof}
\Cref{lem:direct-s12-separated} reduces an arbitrary instance $(G,k)$ to an equivalent instance together with an isolating set $C$ and a partition $V(G)=C\uplus Z\uplus I$ witnessing that $C$ is an isolating set, where $|C|=O(k)$ and $|Z|=O(k^2)$.
Since $S_{1,2}$ is $2$-replaceable witnessed by $(1,1)$ from \Cref{lem:Sd1d2_replaceable}, \Cref{thm:replaceable_kernel} with $(d_1,d_2)=(1,2)$ reduces it to an equivalent instance on $|Z|+2|C|+\binom{|C|}{2}=O(k^2)$ vertices and $|E(C\uplus Z)|+|C|\cdot \left(2|C|+\binom{|C|}{2}\right)=O(k^3)$ edges.
\end{proof}

\begin{proposition}\label{thm:Sd1-result}
For every integer $d_1\geq 1$, \textsc{$S_{d_1,1}$-Packing} admits a polynomial kernel with $O(d_1^6k^2)$ vertices and $O(d_1^7k^2+d_1^6k^3)$ edges.
\end{proposition}
\begin{proof}
Set $d:=d_1+1$.
\Cref{lem:direct-sd11-separated} reduces an arbitrary instance $(G,k)$ to an equivalent instance together with an isolating set $C$ and a partition $V(G)=C\uplus Z\uplus I$ witnessing that $C$ is an isolating set, where $|C|=O\!\left(d_1^2k\right)$ and $|Z|=O\!\left(d_1^6k^2\right)$.
Since $S_{d_1,1}$ is $2$-replaceable witnessed by $(d_1,1)$ from \Cref{lem:Sd1d2_replaceable},
\Cref{thm:replaceable_kernel} applied to the pattern $S_{d_1,1}$ reduces it to an equivalent instance on $|Z|+(d_1+1)|C|+\binom{|C|}{2}=O\!\left(d_1^6k^2\right)$ vertices and $|E(C\uplus Z)|+|C|\cdot \left(2|C|+\binom{|C|}{2}\right)=O\!\left(d_1^7k^2+d_1^6k^3\right)$ edges.
\end{proof}

We now turn to the general subdivided-star case.

\SubdividedStarKernel*
\begin{proof}
\Cref{lem:first-stage-vc} reduces an arbitrary instance $(G,k)$ to an equivalent instance together with a vertex cover $C$ of size at most $O\!\left(d^5k^2\right)$.
Since $S_{d_1,d_2}$ is $2$-replaceable witnessed by $(d_1,1)$ from \Cref{lem:Sd1d2_replaceable},
\Cref{thm:replaceable_kernel} reduces it to an equivalent instance on $(d_1+1)|C|+\binom{|C|}{2}=O\!\left(d^{10}k^4\right)$ vertices and $|C|^2+|C|\cdot \left(2|C|+\binom{|C|}{2}\right)=O\!\left(d^{15}k^6\right)$ edges.
\end{proof}
}

This argument does not extend to the case $d_1=0$.
The first step, namely \Cref{lem:first-stage-vc}, still applies and reduces \textsc{$S_{0,d_2}$-Packing} to an equivalent instance with a vertex cover of size $O(k^2)$.
However, the second step no longer improves the generic bound: by \Cref{prop:S0d-not-replaceable}, the graph $S_{0,d_2}$ is not $\ell$-replaceable for any $\ell<d_2$, so \Cref{thm:replaceable_kernel} yields at best an $O(k^{2d_2})$-vertex kernel, which is no better than the generic \textsc{$d$-Set Packing} bound.

\longonly{
\section{Lower Bound for \textsc{$S_{0,d}$-Packing}}\label{sec:lowerbound}

As explained at the end of the previous section, the two-step approach that yields polynomial kernels for subdivided stars with $d_1\ge 1$ does not improve the generic bound for $S_{0,d}$.
The main result of this section shows that this is not merely a limitation of the method.

\SubdividedStarHardness*

To prove it, fix an integer $d\ge 3$ and let $H:=S_{0,d}$.
We give a parameter-preserving reduction from the following \textsc{Exact $d$-Set Packing}.

\problemdef{Exact $d$-Set Packing:}{Integers $k\ge 0$ and $d\ge 1$, a set $U$ of size $dk$, and a family $\mathcal F\subseteq \binom{U}{d}$.}{Is there a pairwise disjoint subfamily $\mathcal F'\subseteq \mathcal F$ of size $k$?}
Since every set in $\mathcal F$ has size exactly $d$ and $|U|=dk$, such a subfamily exists if and only if it is a partition of $U$.
We will use the following lower bound for this problem.

\begin{lemma}[\cite{DellM12}]
\label{lem:source-lb}
For every fixed integer $d\ge 3$ and $\varepsilon>0$, \textsc{Exact $d$-Set Packing} does not admit a compression of size $O(k^{d-\varepsilon})$ unless $\mathrm{NP}\subseteq \mathrm{coNP}/\mathrm{poly}$.
\end{lemma}

We now reduce \textsc{Exact $d$-Set Packing} to \textsc{$S_{0,d}$-Packing}.
For every element $u\in U$, create two vertices $c_u$ and $i_u$, and let
\(
C:=\{c_u:u\in U\}
\)
and
\(
I:=\{i_u:u\in U\}.
\)
For every set $F\in \mathcal F$, create one vertex $x_F$, and let
\(
X:=\{x_F:F\in \mathcal F\}.
\)
Make $C$, $I$, and $X$ independent.
For every $u\in U$, add the edge $c_ui_u$, so that these edges form a perfect matching between $C$ and $I$.
For every $F\in \mathcal F$ and every $u\in F$, add the edge $x_Fc_u$.
This defines the graph $G$.
\begin{lemma}\label{lem:reduction}
$(U,\mathcal F,k)$ is a yes-instance of \textsc{Exact $d$-Set Packing} if and only if $(G,k)$ is a yes-instance of \textsc{$S_{0,d}$-Packing}.
\end{lemma}

\begin{proof}
Assume first that $(U,\mathcal F,k)$ is a yes-instance.
Let $\mathcal F'=\{F_1,\dots,F_k\}\subseteq \mathcal F$ be a family of $k$ pairwise disjoint sets.
For each $j\in [k]$, consider the vertex $x_{F_j}\in X$.
Since $F_j\in \binom{U}{d}$, the vertex $x_{F_j}$ is adjacent to exactly the $d$ vertices $\{c_u:u\in F_j\}\subseteq C$.
Together with the matched vertices $\{i_u:u\in F_j\}\subseteq I$, these vertices form a copy of $S_{0,d}$ centered at $x_{F_j}$, with branches $x_{F_j}-c_u-i_u$ for $u\in F_j$.
Because the sets $F_1,\dots,F_k$ are pairwise disjoint, the resulting $k$ copies of $S_{0,d}$ are pairwise vertex-disjoint.
Hence $(G,k)$ is a yes-instance.

Assume now that $(G,k)$ is a yes-instance, and let $P_1,\dots,P_k$ be pairwise vertex-disjoint copies of $S_{0,d}$ in $G$.
Since $G$ is bipartite with bipartition $C\uplus (I\uplus X)$, for every $j\in [k]$, one of the following holds.
\begin{itemize}
    \item If the center of $P_j$ lies in $I\uplus X$, then the $d$ internal vertices of its branches lie in $C$.
    Hence $P_j$ uses exactly $d$ vertices of $C$.
    \item If the center of $P_j$ lies in $C$, then the $d$ leaves of its branches lie in $C$.
    Hence $P_j$ uses exactly $d+1$ vertices of $C$.
\end{itemize}
Therefore, every copy $P_j$ uses at least $d$ vertices of $C$.
Since the copies are vertex-disjoint and $|C|=dk$, this forces every $P_j$ to use exactly $d$ vertices of $C$.
In particular, no copy $P_j$ can have its center in $C$.
Also, no copy can have its center in $I$, because every vertex of $I$ has degree $1$ in $G$, while the center of $S_{0,d}$ has degree $d\ge 3$.
Hence, the center of every copy $P_j$ lies in $X$.

Fix $j\in [k]$ and let the center of $P_j$ be $x_F$ for some $F\in \mathcal F$.
The vertex $x_F$ is adjacent exactly to the vertices $\{c_u:u\in F\}\subseteq C$, and it has degree exactly $d$.
Since $P_j$ is a copy of $S_{0,d}$ centered at $x_F$, the $d$ internal vertices of its branches are therefore exactly the vertices $\{c_u:u\in F\}$.
Thus, every copy $P_j$ determines one set $F_j\in \mathcal F$.

Because the copies $P_1,\dots,P_k$ are vertex-disjoint, the corresponding subsets of $C$ are pairwise disjoint.
Hence, the sets $F_1,\dots,F_k$ are pairwise disjoint members of $\mathcal F$.
Since there are $k$ of them and $|U|=dk$, they form a solution to \textsc{Exact $d$-Set Packing}.
\end{proof}

\begin{proof}[Proof of \Cref{thm:S0d-lb-result}]
The reduction from \textsc{Exact $d$-Set Packing} to \textsc{$S_{0,d}$-Packing} given above is polynomial-time and preserves the parameter $k$ by \Cref{lem:reduction}.
Hence any compression for \textsc{$S_{0,d}$-Packing} of size $O(k^{d-\varepsilon})$ would immediately yield a compression of the same size for \textsc{Exact $d$-Set Packing}.
This contradicts \Cref{lem:source-lb}.
\end{proof}
}

\longonly{
\section{Paw-Packing}\label{sec:paw}

\newcommand{\Etriangle}{E_{\triangle}}

In this section, we prove the following theorem.

\PawKernel*

The argument follows the same broad two-step pattern as for subdivided stars.
In the first step, we reduce an arbitrary instance of \textsc{Paw-Packing} to an equivalent one with a vertex cover $C$ of size $O(k^2)$ and an independent set $I:=V(G)\setminus C$.
The crucial additional point is to control the set
\[
\Etriangle(C,I):=\{uv\in E(G[C])\, \colon\, \exists x\in I\text{ with }ux,vx\in E(G)\},
\]
that is, the edges of the cover that can lie on triangles with a vertex of $I$. 
We ensure that $|\Etriangle(C,I)|=O(k^2)$.
In the second step, we reduce instances by keeping only $O(|C|+|\Etriangle(C,I)|)=O(k^2)$ representative vertices on the independent side.

\subsection{Reducing to Instances with a Quadratic-Size Vertex Cover}\label{subsec:paw-stage1}

Let $\mathcal P$ be a maximal family of pairwise vertex-disjoint paws in $G$, and put
\(
C:=\bigcup_{P\in\mathcal P}V(P).
\)
If $|\mathcal P|\ge k$, then $(G,k)$ is already a yes-instance.
Hence we may assume that $|\mathcal P|<k$.
Since each paw has four vertices, this gives $|C|=4|\mathcal P|<4k$, and maximality implies that $G-C$ is paw-free.

The next lemma describes the connected components of the paw-free graph $G-C$.
\begin{lemma}\label{lem:paw-components}
Every connected component of the paw-free graph $G-C$ is either a triangle or triangle-free.
\end{lemma}
\begin{proof}
Let $R$ be a connected component of $G-C$, and suppose that $R$ contains a triangle $T$.
Then no vertex of $R\setminus V(T)$ can be adjacent to a vertex of $T$, because otherwise $T$ together with such a vertex would contain a paw.
Since $R$ is connected, this implies that $R=T$.
\end{proof}

\subsubsection{Reducing triangle components}
We first reduce the number of triangle components of $G-C$ by applying the expansion lemma, and then move the surviving ones into the vertex cover.

Let $\mathcal R_\triangle$ be the family of connected components of $G-C$ that are triangles.
We can safely delete all isolated triangles; thus, we can assume each $R\in \mathcal R$ has a neighbor in $C$.

Define the auxiliary bipartite graph $\mathcal{B}_{\triangle}:=(C\uplus \mathcal R_\triangle,F)$, where a vertex $c\in C$ is adjacent to a component $T\in\mathcal R_\triangle$ if and only if some vertex of $T$ is adjacent to $c$ in $G$.
Apply \Cref{lem:expansion} to $\mathcal{B}_{\triangle}$ with $q=1$.
We obtain subsets $\widehat C\subseteq C$ and $\widehat{\mathcal R}_\triangle\subseteq \mathcal R_\triangle$, and an edge set $M\subseteq F$ such that
\begin{enumerate}
    \item $|\widehat{\mathcal R}_\triangle|\le |\widehat C|$,
    \item every vertex of $C\setminus \widehat C$ is incident in $M$ with exactly one edge,
    \item every vertex of $\mathcal R_\triangle\setminus \widehat{\mathcal R}_\triangle$ is incident in $M$ with at most one edge,
    \item there is no edge of $\mathcal{B}_{\triangle}$ between $\widehat C$ and $\mathcal R_\triangle\setminus \widehat{\mathcal R}_\triangle$.
\end{enumerate}
Set $A:=C\setminus \widehat C$ and $\mathcal D:=\mathcal R_\triangle\setminus \widehat{\mathcal R}_\triangle$.
We apply the following reduction rule.

\begin{rrule}\label{rule:paw-triangle}
If $A\neq\emptyset$, then delete $A\uplus \bigcup_{T\in\mathcal D}V(T)$ and decrease $k$ by $|A|$.
\end{rrule}

\begin{lemma}\label{lem:paw-triangle-safe}
\Cref{rule:paw-triangle} is safe.
\end{lemma}

\begin{proof}
Let $G':=G-\Bigl(A\uplus \bigcup_{T\in\mathcal D}V(T)\Bigr)$ and $k':=k-|A|$.
Suppose first that $(G',k')$ is a yes-instance.
The deleted side contains $|A|$ pairwise vertex-disjoint paws.
Together with any packing of $k'$ paws in $G'$, these yield a packing of $k'+|A|=k$ paws in $G$.

Conversely, suppose that $(G,k)$ is a yes-instance and let $\mathcal Q$ be a packing of $k$ paws in $G$.
If a paw $Q\in\mathcal Q$ uses a vertex of some triangle component $T\in\mathcal D$, $Q$ contains a vertex of $C$ adjacent to $T$.
By the fourth property, no vertex of $\widehat C$ is adjacent to a component of $\mathcal D$, and therefore this vertex of $C$ belongs to $A$.
Thus, every paw crossing $A\uplus \bigcup_{T\in\mathcal D}V(T)$ also crosses $A$.
Since the packing is vertex-disjoint, at most $|A|$ paws of $\mathcal Q$ cross the deleted side.
Removing them leaves at least $k-|A|=k'$ paws in $G'$.
Hence $(G',k')$ is a yes-instance.
\end{proof}

If $|A|\ge k$, we return the disjoint union of $k$ paws.
Otherwise, if $A\neq\emptyset$, we apply \Cref{rule:paw-triangle}.
Consequently, $|\mathcal R_\triangle|=|\widehat{\mathcal R}_\triangle|\le |\widehat C_0|=|C_0|<4k$ by the first property.
Move all surviving triangle components into $C$.
Then $G-C$ is triangle-free by \Cref{lem:paw-components} and $|C|\le 4k+3|\mathcal R_\triangle|<16k$.

\subsubsection{Reducing non-triangle components}
It remains to control the triangle-free part of $G-C$.
The argument has three steps.
First, we exhaust a reduction rule showing that a sufficiently large matching in some neighborhood $N_{G-C}(c)$ already yields one additional paw.
Second, for each $c\in C$, we choose a small set $X_c\subseteq N_{G-C}(c)$ in a specific way: if $N_{G-C}(c)$ is already small we keep it entirely, and otherwise we keep a small vertex cover of $N_{G-C}(c)$.
Third, after enlarging the vertex cover by all sets $X_c$, we delete the remaining edges outside the cover.

\begin{rrule}\label{rule:paw-matching}
If there exists $c\in C$ such that $N_{G-C}(c)$ contains a matching of size at least $4k-2$, then delete $c$ and decrease $k$ by $1$.
\end{rrule}

\begin{lemma}\label{lem:paw-matching-safe}
\Cref{rule:paw-matching} is safe.
\end{lemma}
\begin{proof}
Let $c\in C$, and let $M$ be a matching of size at least $4k-2$ in $N_{G-C}(c)$.
Clearly, if $(G,k)$ is a yes-instance, then $(G-c,k-1)$ is a yes-instance.
We focus on the opposite direction and assume $(G-c,k-1)$ is a yes-instance.

Let $\mathcal Q$ be a packing of $k-1$ paws in $G-c$.
Since each paw has four vertices, the packing uses at most $4(k-1)$ vertices.
Because the edges of $M$ are pairwise vertex-disjoint, at most $4(k-1)$ edges of $M$ intersect $V(\mathcal Q)$.
Hence at least $(4k-2)-4(k-1)=2$ edges of $M$ remain unused in $\mathcal{Q}$.
Choose two such edges, say $xy$ and $zw$.
Then $x,y,z\in N_{G-C}(c)$, and therefore the four vertices $c,x,y,z$ contain a paw: use the triangle on $c,x,y$ and the edge $cz$ as the pendant edge.
This paw is disjoint from every paw of $\mathcal Q$.
Hence $(G,k)$ is a yes-instance.
\end{proof}

Now, for every $c\in C$, $N_{G-C}(c)$ does not have a matching of size $4k-2$.
Hence, $N_{G-C}(c)$ has a vertex cover of size at most $2\cdot (4k-2)<8k$.
For each $c\in C$, define a set $X_c\subseteq N_{G-C}(c)$ as follows:
\begin{itemize}
    \item If $|N_{G-C}(c)|\leq 8k$, set $X_c:=N_{G-C}(c)$.
    \item If $|N_{G-C}(c)| > 8k$, set $X_c$ by a vertex cover of $N_{G-C}(c)$ of size $8k$.
\end{itemize}
Now define
\(
C^{\ast}:=C\cup \bigcup_{c\in C}X_c
\)
and
\(
I^{\ast}:=V(G)\setminus C^{\ast}.
\)
Let $G'$ be obtained from $G$ by deleting every edge with both endpoints in $I^{\ast}$.
We first show that no deleted edge lies in a triangle.
\begin{lemma}\label{lem:paw-no-triangle-on-deleted-edge}
If $yz\in E(G)$ with $y,z\in I^{\ast}$, then $yz$ lies on no triangle of $G$.
\end{lemma}

\begin{proof}
Assume $yzc$ is a triangle.
Since $G-C$ is triangle-free, $c\in C$.
Since $X_c$ is a vertex cover of $N_{G-C}(c)$, either $y$ or $z$ lies in $X_c$, which contradicts $y,z\in I^{\ast}$.
\end{proof}

We now prove that deleting all edges inside $I^{\ast}$ is safe.
The point is that, by the previous lemma, such an edge cannot lie on a triangle.
Hence if a paw uses a deleted edge at all, that edge can only serve as its pendant edge, and we will show that such a paw can be replaced by another paw.

\begin{lemma}\label{lem:paw-delete-Istar-safe}
The instance $(G',k)$ is equivalent to $(G,k)$.
\end{lemma}
\begin{proof}
If $(G',k)$ is a yes-instance, then $(G,k)$ is a yes-instance because $G'$ is a spanning subgraph of $G$.
For the converse direction, suppose that $(G,k)$ is a yes-instance.
Call a paw in $G$ \emph{good} if all its edges are present in $G'$, and \emph{bad} otherwise.

Let $\mathcal Q$ be a packing of $k$ paws in $G$ with the minimum number of bad paws, and assume for contradiction that $\mathcal Q$ contains a bad paw $P$.
Then $P$ contains an edge with both endpoints in $I^{\ast}$.
By \Cref{lem:paw-no-triangle-on-deleted-edge}, such an edge cannot lie on the triangle of $P$.
Since a paw has only one non-triangle edge, $P$ has exactly one deleted edge, and this edge is the pendant edge of $P$.
Therefore, $P$ has the form $V(P)=\{c,b,x,w\}$, where $c\in C$, $b\in C^{\ast}$, the vertices $c,b,x$ form a triangle, and $xw$ is an edge of $I^{\ast}$.
Particularly, $c$ has a neighbor $x\not \in X_{c}$ in $N_{G-C}(c)$; thus, by definition of $X_{c}$, $|N_{G-C}(c)| > 8k$ and $|X_{c}|=8k$.
Since $\mathcal{Q}\setminus \{P\}$ uses $4(k-1) < 8k$ vertices, there is a vertex $y\in X_c$ that is unused in $\mathcal{Q}$.
Define $P'$ to be the paw on the four vertices $\{c,b,x,y\}$, using the triangle $cbx$ and the pendant edge $cy$.
Then, $P'$ is a good triangle disjoint from all paws in $\mathcal{Q}\setminus \{P\}$.
Replacing $P$ by $P'$ decreases the number of bad paws, which contradicts the choice of $\mathcal Q$.
\end{proof}

Since every edge of $G'$ with both endpoints outside $C^{\ast}$ has been deleted, $C^{\ast}$ is a vertex cover of $G'$.
Moreover, $|C^{\ast}|\le |C|+\sum_{c\in C}|X_c|<16k+16k\cdot 8k=O(k^2)$.
It remains to bound the number of edges in $\Etriangle(C^{\ast},I^{\ast})$.

\begin{lemma}\label{lem:paw-active-bound}
$|\Etriangle(C^{\ast},I^{\ast})|= O(k^2)$.
\end{lemma}

\begin{proof}
Since $G-C$ is triangle-free, for each edge in $\Etriangle(C^{\ast},I^{\ast})$, at least one of its endpoints is in $C$.
We split into two cases:

\proofsubparagraph*{Case 1: Both endpoints lie in $C$.}
There are at most $\binom{|C|}{2}=O(k^2)$ such edges because $|C|<16k$.

\proofsubparagraph*{Case 2: Exactly one endpoint lies in $C$.}
Let $cv\in \Etriangle(C^{\ast},I^{\ast})$ with $c\in C$ and $v\in C^{\ast}\setminus C$, witnessed by a vertex $x\in I^{\ast}$.
Then, $v,x\in N_{G-C}(c)$, and thus, either $v$ or $x$ is in $X_c$.
Since $x\in I^{\ast}$, we have $v\in X_c$.
Particularly, the number of such edges is at most $\sum_{c\in C}|X_c|\leq |C|\cdot 8k= O(k^2)$.
Combining the two cases completes the proof.
\end{proof}

The graph $G'$ together with the partition $V(G')=C^{\ast}\uplus I^{\ast}$ satisfies all required properties.
Thus, we have the following.
\begin{lemma}\label{lem:paw-stage1}
Given an instance $(G,k)$ of \textup{\textsc{Paw-Packing}}, one can compute in polynomial time an equivalent instance $(G',k')$ together with a partition $V(G')=C\uplus I$ such that $I$ is independent and $k'\leq k$, $|C|=O(k^2)$, and $|\Etriangle(C,I)|=O(k^2)$.
\end{lemma}

\subsection{Kernelization from a Vertex Cover}\label{subsec:paw-stage2}

Here, we kernelize a bounded-vertex-cover instance obtained in the first stage.
Let $V(G)=C\uplus I$ and assume $C$ is a vertex cover of $G$, $|C|= O(k^2)$, and $|\Etriangle(C,I)|= O(k^2)$.
We abuse notation so that $\Etriangle$ represents $\Etriangle(C,I)$.
Define the bipartite graph $\mathcal{B}:=(L\uplus I,F)$ as follows.
The left side is $L:=C\uplus \Etriangle$.
For $c\in C$ and $x\in I$, add the edge $cx$ to $F$ if $cx\in E(G)$.
For $e=uv\in \Etriangle$ and $x\in I$, add the edge $ex$ to $F$ if $ux,vx\in E(G)$, i.e., $uvx$ is a triangle.

Apply \Cref{lem:expansion} to $\mathcal{B}$ with $q=1$.
We obtain subsets $\widehat L\subseteq L$ and $\widehat I\subseteq I$, and an edge set $M\subseteq F$ satisfying the four properties guaranteed by \Cref{lem:expansion}.
Set $A:=L\setminus \widehat L$ and $D:=I\setminus \widehat I$, and let $R_D\subseteq D$ be the set of vertices of $D$ incident with $M$.
Since $q=1$, the set $M$ is a matching that saturates $A$ into $R_D$.
For every $\lambda\in A$, let $m_{\lambda}\in R_D$ denote the unique vertex such that $\lambda m_{\lambda}\in M$.
Let $G':=G-(D\setminus R_D)$.
From now on, we prove that $(G',k)$ is a kernel of $(G,k)$.
First, we bound the size.
\begin{lemma}\label{lem:paw-size}
The graph $G'$ satisfies
\[
|V(G')|\le 2|C|+|\Etriangle|,
\qquad
|E(G')|\le \binom{|C|}{2}+|C|\bigl(|C|+|\Etriangle|\bigr).
\]
In particular, $|V(G')|=O(k^2)$ and $|E(G')|=O(k^4)$.
\end{lemma}
\begin{proof}
The first property yields $|\widehat I|\le |\widehat L|$.
Since $M$ is a matching that saturates $A$ into $R_D$, $|R_D|=|A|=|L|-|\widehat L|$.
Therefore $|V(G')\cap I|=|\widehat I|+|R_D|\le |\widehat L|+|L|-|\widehat L|=|L|$.
Now $|L|=|C|+|\Etriangle|$, so $|V(G')|=|C|+|V(G')\cap I|\le 2|C|+|\Etriangle|$.
Since $C$ is a vertex cover of $G$, it is also a vertex cover of $G'$.
Hence $|E(G')|\le |E(G'[C])|+|C|\cdot |V(G')\cap I|\le \binom{|C|}{2}+|C|\bigl(|C|+|\Etriangle|\bigr)$.
The final asymptotic bounds follow from the assumptions on $|C|$ and $|\Etriangle|$.
\end{proof}

Next, we prove correctness.
Since $G'$ is a subgraph of $G$, if $(G',k)$ is a yes-instance, $(G,k)$ is a yes-instance.
Thus, we focus on the opposite direction and assume $(G,k)$ is a yes-instance.

Fix a packing $\mathcal P$ of $k$ copies of paws in $G$.
We replace each original copy $P\in \mathcal{P}$ by a copy $P^{\star}$ in $G'$.
The replacement will satisfy three conditions:
\begin{enumerate}
    \item $V(P^{\star})\cap C\subseteq V(P)\cap C$,
    \item $V(P^{\star})\cap \widehat I\subseteq V(P)\cap \widehat I$, and
    \item every vertex of $V(P^{\star})\cap R_D$ is equal to $m_{\lambda}$ for some $\lambda\in (V(P)\cap C)\cup (E(P)\cap \Etriangle)$.
\end{enumerate}

We first explain why these conditions are sufficient.
\begin{lemma}\label{lem:paw-replacement-principle}
Let $\mathcal P$ be a collection of pairwise vertex-disjoint copies of paws in $G$.
Assume that for every $P\in \mathcal P$ there exists a copy $P^{\star}$ of a paw in $G'$ that satisfies the three conditions above.
Then $\{P^{\star}:P\in \mathcal P\}$ is a collection of pairwise vertex-disjoint copies of paws in $G'$.
\end{lemma}
\begin{proof}
Let $P,Q\in \mathcal P$ be distinct.
If a vertex of $P^{\star}\cap Q^{\star}$ lies in $C$, then by the first condition it belongs to both $V(P)$ and $V(Q)$, contrary to the disjointness of $\mathcal P$.
The same argument with the second condition shows that $P^{\star}$ and $Q^{\star}$ are disjoint on $\widehat I$.

Now let $x\in P^{\star}\cap Q^{\star}\cap R_D$.
Since $x\in R_D$, the third property implies that $x$ is incident with exactly one edge of the matching $M$.
Hence there is a unique $\lambda\in A$ with $x=m_{\lambda}$.
By the third condition for $P^{\star}$ and $Q^{\star}$, this label $\lambda$ belongs to both $(V(P)\cap C)\cup (E(P)\cap \Etriangle)$ and $(V(Q)\cap C)\cup (E(Q)\cap \Etriangle)$.
If $\lambda\in C$, then $\lambda\in V(P)\cap V(Q)$, contradicting the disjointness of $\mathcal P$.
If $\lambda=uv\in \Etriangle$, then $u,v\in V(P)\cap C$ and $u,v\in V(Q)\cap C$, again contradicting the disjointness of $\mathcal P$.
Therefore $P^{\star}$ and $Q^{\star}$ are disjoint on $R_D$ as well.
\end{proof}

It remains to construct the replacement.
\begin{lemma}\label{lem:replace-paw}
Let $P$ be a copy of a paw in $G$.
Then there exists a copy $P^{\star}$ of a paw in $G'$ such that:
\begin{enumerate}
    \item $V(P^{\star})\cap C\subseteq V(P)\cap C$,
    \item $V(P^{\star})\cap \widehat I\subseteq V(P)\cap \widehat I$, and
    \item every vertex of $V(P^{\star})\cap R_D$ is equal to $m_{\lambda}$ for some $\lambda\in (V(P)\cap C)\cup (E(P)\cap \Etriangle)$.
\end{enumerate}
\end{lemma}
\begin{proof}
Let $X:=V(P)\cap D$.
We distinguish two cases.

\proofsubparagraph*{Case 1. $X=\{x\}$, where $x$ is the leaf of $P$.}

Let $V(P)=\{a,b,c,x\}$, where $c\in C$ be the unique neighbor of $x$ in $P$.
Since $x\in D$, the fourth property ensures $c\in A$; thus, $m_c\in R_D$ is defined.
Then $V(P^{\star}):=\{a,b,c,m_c\}$ contains a paw with the triangle $abc$ and pendant $cm_c$.

The first two conditions are immediate.
The third condition follows from $c\in C$.

\proofsubparagraph*{Case 2. Some vertex of $X$ lies on the triangle of $P$.}

Let $V(P)=\{a,b,x,t\}$, where $abx$ is a triangle and $x\in X$.
Since $x\in I$ and $abx$ is a triangle, $a,b\in C$, and thus, $ab\in \Etriangle$.
Since $x\in D$, the fourth property ensures $ab\in A$; thus, $m_{ab}\in R_D$ is defined.
Moreover, again from $x\in D$ and the fourth property, $a\in A$; thus, $m_a\in R_D$ is defined.
Then $V(P^{\star}):=\{a,b,m_{ab},m_a\}$ contains a paw with the triangle $abm_{ab}$ and pendant $am_a$.

The first two conditions are immediate.
The third condition follows from $ab\in \Etriangle$ and $a\in C$.
\end{proof}

Combining \Cref{lem:paw-size,lem:paw-replacement-principle,lem:replace-paw} immediately yields the following.
\begin{lemma}\label{lem:paw-stage2}
There is a polynomial-time algorithm that, given an instance $(G,k)$ together with a partition $V(G)=C\uplus I$ such that $C$ is a vertex cover of $G$, $|C|= O(k^2)$, and $|\Etriangle(C,I)|= O(k^2)$, outputs an equivalent instance $(G',k)$ such that
\[
|V(G')|\le 2|C|+|\Etriangle(C,I)|,
\qquad
|E(G')|\le \binom{|C|}{2}+|C|\bigl(|C|+|\Etriangle(C,I)|\bigr).
\]
In particular, $(G',k)$ has $O(k^2)$ vertices and $O(k^4)$ edges.
\end{lemma}
\begin{proof}[Proof of \Cref{thm:paw-result}]
\Cref{lem:paw-stage1} reduces an arbitrary instance $(G,k)$ to an equivalent instance $(G',k')$ with $k'\le k$ together with a partition $V(G')=C\uplus I$ such that $C$ is a vertex cover of size $O(k^2)$ and $|\Etriangle(C,I)|=O(k^2)$.
Applying \Cref{lem:paw-stage2} to this instance yields an equivalent instance on $O(k^2)$ vertices.
\end{proof}
}

\section{Conclusion}

In this paper, we obtained kernels for \textsc{$H$-Packing} beating the generic \textsc{$d$-Set Packing} bound for several new pattern graphs $H$.
In particular, we proved:
\begin{itemize}
    \item for $H=P_5$ and $H=S_{1,2}$, a kernel with $O(k^2)$ vertices and $O(k^3)$ edges,
    \item for $H=S_{d_1,1}$, a kernel with $O(k^2)$ vertices and $O(k^3)$ edges,
    \item for $H=S_{d_1,d_2}$ with $d_1\ge 1$, a kernel with $O(k^4)$ vertices and $O(k^6)$ edges, and
    \item for the paw, a kernel with $O(k^2)$ vertices and $O(k^4)$ edges.
\end{itemize}
By contrast, we did not obtain a smaller kernel for the line $H=S_{0,d}$.
Instead, we proved that for every $d\ge 3$ and every $\varepsilon>0$, \textsc{$S_{0,d}$-Packing} does not admit a compression of size $O(k^{d-\varepsilon})$ unless $\NP\subseteq \coNP/\poly$.
In particular, this shows that kernelization for \textsc{$H$-Packing} is not monotone under taking induced subgraphs.
This suggests that the kernel complexity of \textsc{$H$-Packing} is rather subtle, and that further investigation is needed.
\shortonly{Natural next targets include small cyclic patterns such as $C_4$ and diamond, subdivided stars with longer branches, and small trees that are not subdivided stars.}
\longonly{
Natural next targets for small kernels include:
\begin{itemize}
    \item Small graphs with cycles such as $C_4$ (\Cref{fig:c4}) and diamond (\Cref{fig:diamond}).
    \item Subdivided stars with branches of length $\geq 3$; the smallest example is the six-vertex subdivided star with branch lengths $1,1,3$ (\Cref{fig:s201}).
    \item Small trees that are not subdivided stars; the smallest example is $H$ (\Cref{fig:H}).
\end{itemize}
 
For \textsc{$P_4$-Packing} and \textsc{$P_5$-Packing}, Dell and Marx~\cite{DellM12} asked whether there is a kernel of $O(k^2)$ edges, which matches their lower bound.
Our results partially answer this question by giving an $O(k^2)$-vertex kernel for \textsc{$P_5$-Packing}, but whether the number of edges can also be reduced to $O(k^2)$ remains open.
Moreover, by \Cref{subsec:direct-p5}, it suffices to resolve this question for instances with vertex cover number $O(k)$.
In this sense, the present paper may be viewed as a substantial step toward that goal.

Another promising direction is kernelization of \textsc{Induced $H$-Packing}.
At present, a non-trivially small kernel is known only for the case $H=P_3$, which admits an $O(k)$-vertex kernel~\cite{BessyBTW23}.
It is therefore natural to ask whether there is a four-vertex graph $H$ for which \textsc{Induced $H$-Packing} admits a kernel with a subcubic number of vertices.
}

\longonly{
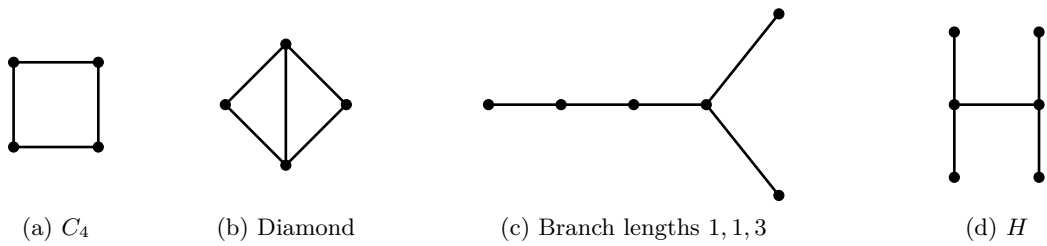
\begin{figure}[t]
\centering
\begin{tikzpicture}[
    baseline=(current bounding box.center),
    every node/.style={inner sep=0pt, outer sep=0pt},
    line cap=round, line join=round,
    ink/.style={draw=black, line width=1.0pt},
    dot/.style={fill=black, draw=black},
    panelcaption/.style={font=\small, align=center}
]
\def\captionY{-2.05}

\begin{scope}[xshift=0cm, x=0.8cm, y=0.8cm]
    \def\r{1.9pt}
    \coordinate (A) at (0,-0.7);
    \coordinate (B) at (1.4,-0.7);
    \coordinate (C) at (1.4,0.7);
    \coordinate (D) at (0,0.7);
    \draw[ink] (A)--(B)--(C)--(D)--cycle;
    \foreach \p in {A,B,C,D}{\fill[dot] (\p) circle (\r);}
    \node[panelcaption] at (0.7,\captionY) {(a) $C_4$};
\end{scope}
\phantomsubcaption\label{fig:c4}

\begin{scope}[xshift=3.6cm, x=0.8cm, y=0.8cm]
    \def\r{1.9pt}
    \coordinate (L) at (-1,0);
    \coordinate (T) at (0,1);
    \coordinate (R) at (1,0);
    \coordinate (B) at (0,-1);
    \draw[ink] (L)--(T)--(R)--(B)--cycle;
    \draw[ink] (T)--(B);
    \foreach \p in {L,T,R,B}{\fill[dot] (\p) circle (\r);}
    \node[panelcaption] at (0,\captionY) {(b) Diamond};
\end{scope}
\phantomsubcaption\label{fig:diamond}

\begin{scope}[xshift=8.2cm, x=1cm, y=0.8cm]
    \def\r{1.9pt}
    \begin{scope}[rotate=90]
        \coordinate (O) at (0,-1.2);
        \coordinate (L) at (-1.2,-2.4);
        \coordinate (R) at (1.2,-2.4);
        \coordinate (A) at (0,0);
        \coordinate (B) at (0,1.2);
        \coordinate (C) at (0,2.4);
        \draw[ink] (O) -- (L);
        \draw[ink] (O) -- (R);
        \draw[ink] (O) -- (A) -- (B) -- (C);
        \foreach \p in {O,L,R,A,B,C}{\fill[dot] (\p) circle (\r);}
    \end{scope}
    \node[panelcaption] at (0,\captionY) {(c) Branch lengths $1,1,3$};
\end{scope}
\phantomsubcaption\label{fig:s201}

\begin{scope}[xshift=13.0cm, x=0.7cm, y=0.8cm]
    \def\r{1.9pt}
    \coordinate (T1) at (-0.8,1.2);
    \coordinate (M1) at (-0.8,0);
    \coordinate (B1) at (-0.8,-1.2);
    \coordinate (T2) at (0.8,1.2);
    \coordinate (M2) at (0.8,0);
    \coordinate (B2) at (0.8,-1.2);
    \draw[ink] (T1)--(M1)--(B1);
    \draw[ink] (T2)--(M2)--(B2);
    \draw[ink] (M1)--(M2);
    \foreach \p in {T1,M1,B1,T2,M2,B2}{\fill[dot] (\p) circle (\r);}
    \node[panelcaption] at (0,\captionY) {(d) $H$};
\end{scope}
\phantomsubcaption\label{fig:H}
\end{tikzpicture}

\caption{Candidate pattern graphs for future work.}
\end{figure}
}

\bibliography{ref}

@inproceedings{BessyBTW23,
  author       = {St{\'{e}}phane Bessy and
                  Marin Bougeret and
                  Dimitrios M. Thilikos and
                  Sebastian Wiederrecht},
  title        = {{Kernelization for Graph Packing Problems via Rainbow Matching}},
  booktitle    = {Proceedings of the Annual {ACM-SIAM} Symposium on Discrete
                  Algorithms ({SODA} 2023)},
  pages        = {3654--3663},
  publisher    = {{SIAM}},
  year         = {2023},
  doi          = {10.1137/1.9781611977554.CH139}
}

@inproceedings{DellM12,
  author       = {Holger Dell and
                  D{\'{a}}niel Marx},
  title        = {{Kernelization of Packing Problems}},
  booktitle    = {Proceedings of the Annual {ACM-SIAM} Symposium on Discrete
                  Algorithms ({SODA} 2012)},
  pages        = {68--81},
  publisher    = {{SIAM}},
  year         = {2012},
  doi          = {10.1137/1.9781611973099.6},
  note         = {arXiv version available at http://arxiv.org/abs/1812.03155}
}

@article{FominLLSTZ19,
  author       = {Fedor V. Fomin and
                  Tien{-}Nam Le and
                  Daniel Lokshtanov and
                  Saket Saurabh and
                  St{\'{e}}phan Thomass{\'{e}} and
                  Meirav Zehavi},
  title        = {Subquadratic Kernels for Implicit 3-Hitting Set and 3-Set Packing
                  Problems},
  journal      = {{ACM} Trans. Algorithms},
  volume       = {15},
  number       = {1},
  pages        = {13:1--13:44},
  year         = {2019},
  doi          = {10.1145/3293466}
}

@article{Thomasse10,
  author       = {St{\'{e}}phan Thomass{\'{e}}},
  title        = {A $4k^2$ Kernel for Feedback Vertex Set},
  journal      = {{ACM} Trans. Algorithms},
  volume       = {6},
  number       = {2},
  pages        = {32:1--32:8},
  year         = {2010},
  doi          = {10.1145/1721837.1721848}
}

@article{Xiao17,
  author       = {Mingyu Xiao},
  title        = {On a generalization of {N}emhauser and {T}rotter's local optimization theorem},
  journal      = {J. Comput. Syst. Sci.},
  volume       = {84},
  pages        = {97--106},
  year         = {2017},
  doi          = {10.1016/J.JCSS.2016.08.003}
}

@article{CervenyCS24,
  author       = {Radovan Cerven{\'{y}} and
                  Pratibha Choudhary and
                  Ondrej Such{\'{y}}},
  title        = {On kernels for $d$-path vertex cover},
  journal      = {J. Comput. Syst. Sci.},
  volume       = {144},
  pages        = {103531},
  year         = {2024},
  doi          = {10.1016/J.JCSS.2024.103531}
}

@article{Abu-KhzamSet10,
  author       = {Faisal N. Abu{-}Khzam},
  title        = {An improved kernelization algorithm for $r$-Set Packing},
  journal      = {Inf. Process. Lett.},
  volume       = {110},
  number       = {16},
  pages        = {621--624},
  year         = {2010},
  doi          = {10.1016/j.ipl.2010.04.020}
}

@article{BjorklundHKK17,
  author       = {Andreas Bj{\"o}rklund and
                  Thore Husfeldt and
                  Petteri Kaski and
                  Mikko Koivisto},
  title        = {Narrow sieves for parameterized paths and packings},
  journal      = {J. Comput. Syst. Sci.},
  volume       = {87},
  pages        = {119--139},
  year         = {2017},
  doi          = {10.1016/J.JCSS.2017.03.003}
}

@inproceedings{Nederlof25,
  author       = {Jesper Nederlof},
  title        = {{Weighted $k$-Path and Other Problems in Almost $O^*(2^k)$ Deterministic Time via Dynamic Representative Sets}},
  booktitle    = {Proceedings of the {IEEE} Annual Symposium on Foundations of Computer Science ({FOCS} 2025)},
  pages        = {2813--2824},
  publisher    = {{IEEE} Computer Society},
  year         = {2025},
  doi          = {10.1109/FOCS63196.2025.00143}
}

@article{FengWC14,
  author       = {Qilong Feng and
                  Jianxin Wang and
                  Jianer Chen},
  title        = {Matching and Weighted ${P}_{2}$-Packing: Algorithms and Kernels},
  journal      = {Theor. Comput. Sci.},
  volume       = {522},
  pages        = {85--94},
  year         = {2014},
  doi          = {10.1016/J.TCS.2013.12.011}
}

@inproceedings{Zehavi15,
  author       = {Meirav Zehavi},
  title        = {{Mixing Color Coding-Related Techniques}},
  booktitle    = {Proceedings of the 23rd Annual European Symposium on Algorithms ({ESA} 2015)},
  series       = {Lecture Notes in Computer Science},
  volume       = {9294},
  pages        = {1037--1049},
  publisher    = {Springer},
  year         = {2015},
  doi          = {10.1007/978-3-662-48350-3\_86}
}

@article{FellowsKNRRSTW08,
  author       = {Michael R. Fellows and
                  Christian Knauer and
                  Naomi Nishimura and
                  Prabhakar Ragde and
                  Frances A. Rosamond and
                  Ulrike Stege and
                  Dimitrios M. Thilikos and
                  Sue Whitesides},
  title        = {Faster Fixed-Parameter Tractable Algorithms for Matching and Packing
                  Problems},
  journal      = {Algorithmica},
  volume       = {52},
  number       = {2},
  pages        = {167--176},
  year         = {2008},
  doi          = {10.1007/S00453-007-9146-Y}
}

@article{PrietoS06,
  author       = {Elena Prieto{-}Rodriguez and
                  Christian Sloper},
  title        = {Looking at the stars},
  journal      = {Theor. Comput. Sci.},
  volume       = {351},
  number       = {3},
  pages        = {437--445},
  year         = {2006},
  doi          = {10.1016/J.TCS.2005.10.009}
}

@inproceedings{Moser09,
  author       = {Hannes Moser},
  editor       = {Mogens Nielsen and
                  Anton{\'{\i}}n Kucera and
                  Peter Bro Miltersen and
                  Catuscia Palamidessi and
                  Petr Tuma and
                  Frank D. Valencia},
  title        = {{A Problem Kernelization for Graph Packing}},
  booktitle    = {Proceedings of the 35th International Conference on Current Trends
                  in Theory and Practice of Computer Science ({SOFSEM} 2009)},
  series       = {Lecture Notes in Computer Science},
  volume       = {5404},
  pages        = {401--412},
  publisher    = {Springer},
  year         = {2009},
  doi          = {10.1007/978-3-540-95891-8\_37}
}

@article{FellowsGMN11,
  author       = {Michael R. Fellows and
                  Jiong Guo and
                  Hannes Moser and
                  Rolf Niedermeier},
  title        = {A generalization of {N}emhauser and {T}rotter's local optimization
                  theorem},
  journal      = {J. Comput. Syst. Sci.},
  volume       = {77},
  number       = {6},
  pages        = {1141--1158},
  year         = {2011},
  doi          = {10.1016/J.JCSS.2010.12.001}
}

@article{WangNFC10,
  author       = {Jianxin Wang and
                  Dan Ning and
                  Qilong Feng and
                  Jianer Chen},
  title        = {An improved kernelization for ${P}_{2}$-packing},
  journal      = {Inf. Process. Lett.},
  volume       = {110},
  number       = {5},
  pages        = {188--192},
  year         = {2010},
  doi          = {10.1016/J.IPL.2009.12.002}
}

@article{ChenFSWY19,
  author       = {Jianer Chen and
                  Henning Fernau and
                  Peter Shaw and
                  Jianxin Wang and
                  Zhibiao Yang},
  title        = {Kernels for packing and covering problems},
  journal      = {Theor. Comput. Sci.},
  volume       = {790},
  pages        = {152--166},
  year         = {2019},
  doi          = {10.1016/J.TCS.2019.04.018}
}

@article{LiYC22,
  author       = {Wenjun Li and
                  Junjie Ye and
                  Yixin Cao},
  title        = {A $5k$-vertex kernel for ${P}_2$-packing},
  journal      = {Theor. Comput. Sci.},
  volume       = {910},
  pages        = {1--13},
  year         = {2022},
  doi          = {10.1016/J.TCS.2022.01.032}
}

@inproceedings{FellowsHRST04,
  author       = {Mike Fellows and
                  Pinar Heggernes and
                  Frances A. Rosamond and
                  Christian Sloper and
                  Jan Arne Telle},
  editor       = {Juraj Hromkovic and
                  Manfred Nagl and
                  Bernhard Westfechtel},
  title        = {{Finding $k$ Disjoint Triangles in an Arbitrary Graph}},
  booktitle    = {Proceedings of the 30th International Workshop on Graph-Theoretic
                  Concepts in Computer Science ({WG} 2004)},
  series       = {Lecture Notes in Computer Science},
  volume       = {3353},
  pages        = {235--244},
  publisher    = {Springer},
  year         = {2004},
  doi          = {10.1007/978-3-540-30559-0\_20}
}

@article{LokshtanovMS18,
  author       = {Daniel Lokshtanov and
                  D{\'{a}}niel Marx and
                  Saket Saurabh},
  title        = {Known Algorithms on Graphs of Bounded Treewidth Are Probably Optimal},
  journal      = {{ACM} Trans. Algorithms},
  volume       = {14},
  number       = {2},
  pages        = {13:1--13:30},
  year         = {2018},
  doi          = {10.1145/3170442}
}

@book{Cygan2015,
  author       = {Marek Cygan and
                  Fedor V. Fomin and
                  Lukasz Kowalik and
                  Daniel Lokshtanov and
                  D{\'{a}}niel Marx and
                  Marcin Pilipczuk and
                  Michal Pilipczuk and
                  Saket Saurabh},
  title        = {Parameterized Algorithms},
  publisher    = {Springer},
  year         = {2015},
  doi          = {10.1007/978-3-319-21275-3},
  isbn         = {978-3-319-21274-6}
}

@book{FominLSZ19,
  author       = {Fedor V. Fomin and
                  Daniel Lokshtanov and
                  Saket Saurabh and
                  Meirav Zehavi},
  title        = {Kernelization: Theory of Parameterized Preprocessing},
  publisher    = {Cambridge University Press},
  year         = {2019},
  doi          = {10.1017/9781107415157},
  isbn         = {9781107057760}
}

@article{KirkpatrickH83,
  author       = {David G. Kirkpatrick and
                  Pavol Hell},
  title        = {On the Complexity of General Graph Factor Problems},
  journal      = {{SIAM} J. Comput.},
  volume       = {12},
  number       = {3},
  pages        = {601--609},
  year         = {1983},
  doi          = {10.1137/0212040},
  bibsource    = {dblp computer science bibliography, https://dblp.org}
}

@inproceedings{KirkpatrickH78,
  author       = {David G. Kirkpatrick and
                  Pavol Hell},
  title        = {{On the Completeness of a Generalized Matching Problem}},
  booktitle    = {Proceedings of the Annual {ACM} Symposium on Theory of Computing ({STOC} 1978)},
  pages        = {240--245},
  publisher    = {{ACM}},
  year         = {1978},
  doi          = {10.1145/800133.804353}
}

@article{Yuster07,
  author       = {Raphael Yuster},
  title        = {Combinatorial and computational aspects of graph packing and graph
                  decomposition},
  journal      = {Comput. Sci. Rev.},
  volume       = {1},
  number       = {1},
  pages        = {12--26},
  year         = {2007},
  doi          = {10.1016/J.COSREV.2007.07.002},
  bibsource    = {dblp computer science bibliography, https://dblp.org}
}
\end{document}